\documentclass[a4paper,UKenglish]{lipics-v2016}
\pdfoutput=1

\usepackage{url}
\usepackage{listings}
\usepackage{hyperref}

\usepackage{tikz}
\usetikzlibrary{decorations.pathreplacing}  
\usepackage{etex}
\usepackage{booktabs}
\usepackage{color}
\usepackage{galois}
\usepackage{paralist}
\usepackage{multirow}
\usepackage{tabularx}
\usepackage{multicol}
\usepackage{comment}
\usepackage{xspace}
\usepackage{float}
\usepackage{wrapfig}
\usepackage{amssymb}
\usepackage{amsthm}
\usepackage{pdflscape}
\usepackage{paralist}

\definecolor{light-gray}{gray}{0.80}

\newcommand{\ourttfontsize}{\normalsize}

\newcommand{\op}[1]{\texttt{#1}\xspace}

\newcommand{\deqName}{\texttt{\ourttfontsize{}deq}\xspace}
\newcommand{\enqName}{\texttt{\ourttfontsize{}enq}\xspace}
\newcommand{\deq}[1]{\ensuremath \deqName(#1)\xspace}
\newcommand{\enq}[1]{\ensuremath \enqName(#1)\xspace}
\newcommand{\emptyValue}{\texttt{\ourttfontsize{}empty}\xspace}

\newcommand{\insName}{\texttt{\ourttfontsize{}ins}\xspace}
\newcommand{\remName}{\texttt{\ourttfontsize{}rem}\xspace}
\newcommand{\ins}[1]{\ensuremath \insName(#1)\xspace}
\newcommand{\rem}[1]{\ensuremath \remName(#1)\xspace}
\newcommand{\headName}{\texttt{\ourttfontsize{}head}\xspace}
\newcommand{\sizeName}{\texttt{\ourttfontsize{}size}\xspace}
\newcommand{\peekName}{\texttt{\ourttfontsize{}peek}\xspace}
\newcommand{\topName}{\texttt{\ourttfontsize{}top}\xspace}
\newcommand{\head}[1]{\ensuremath \headName(#1)\xspace}
\newcommand{\topp}[1]{\ensuremath \topName(#1)\xspace}
\newcommand{\size}[1]{\ensuremath \sizeName(#1)\xspace}
\newcommand{\peek}[1]{\ensuremath \peekName(#1)\xspace}
\newcommand{\Emp}{\texttt{\ourttfontsize{}Emp}\xspace}

\newcommand{\iName}{\texttt{i}\xspace}
\newcommand{\rName}{\texttt{r}\xspace}
\newcommand{\ii}[1]{\ensuremath \iName(#1)\xspace}
\newcommand{\rr}[1]{\ensuremath \rName(#1)\xspace}

\newcommand{\emptyName}{\texttt{\ourttfontsize{}empty}\xspace}
\newcommand{\eempty}[1]{\emptyName(#1)\xspace}
\newcommand{\pushName}{\texttt{\ourttfontsize{}push}\xspace}
\newcommand{\popName}{\texttt{\ourttfontsize{}pop}\xspace}
\newcommand{\push}[1]{\ensuremath \pushName(#1)\xspace}
\newcommand{\pop}[1]{\ensuremath \popName(#1)\xspace}
\newcommand{\Ins}{\ensuremath\texttt{\ourttfontsize{}Ins}}
\newcommand{\Rem}{\ensuremath\texttt{\ourttfontsize{}Rem}}
\newcommand{\DOb}{\ensuremath\texttt{\ourttfontsize{}DOb}}
\newcommand{\SOb}{\ensuremath\texttt{\ourttfontsize{}SOb}}

\newcommand{\abOrder}[1]{\ensuremath{\prec_{#1}}}
\newcommand{\precOrder}[1]{\ensuremath{<_{#1}}}

\newcommand{\history}[1]{\ensuremath{\mathbf{#1}}\xspace}
\newcommand{\Complete}[1]{\ensuremath{\rm{\texttt{Complete({#1})}}}}

\newcommand{\crazy}{\textsf{\scriptsize NearlyQ}}

\newtheorem{proposition}{Proposition}

\newcommand{\qedd}{\hfill\ensuremath{\diamond}}

\newcommand{\OurTheorem}[3]{\bigbreak\noindent\textbf{Theorem~#1 (#2).} {\it #3}}
\newcommand{\OurProposition}[3]{\bigbreak\noindent\textbf{Proposition~#1 (#2).} {\it #3}}

\begin{document}

\title{Local Linearizability\footnote{This paper is an extended version of~\cite{CONCUR16}}}

\author[1]{Andreas~Haas}
\author[2]{Thomas~A.~Henzinger}
\author[3]{Andreas~Holzer}
\author[4]{Christoph~M.~Kirsch}
\author[1]{Michael~Lippautz}
\author[1]{Hannes~Payer}
\author[5]{Ali~Sezgin}
\author[4]{Ana~Sokolova}
\author[6,7]{Helmut~Veith}
\authorrunning{A.~Haas et al.}
\affil[1]{Google Inc.}
\affil[2]{IST Austria, Austria}
\affil[3]{University of Toronto, Canada}
\affil[4]{University of Salzburg, Austria}
\affil[5]{University of Cambridge, UK}
\affil[6]{Vienna University of Technology, Austria}
\affil[7]{\rm{\it{Forever in our hearts}}}



\maketitle

\begin{abstract}
The semantics of concurrent data structures is usually given by a sequential 
specification and a consistency condition. 
Linearizability is the most popular consistency condition due to its 
simplicity and general applicability.
Nevertheless, for applications that do not require all guarantees offered by 
linearizability, recent research has focused on improving performance and
scalability of concurrent data structures by relaxing their semantics. 

In this paper, we present local linearizability, a relaxed consistency
condition that is applicable to \emph{container-type} concurrent data 
structures like pools, queues, and stacks.
While linearizability requires that the effect of each operation is observed 
by all threads at the same time, local linearizability only requires that for 
each thread~T, the effects of its local insertion operations and the effects 
of those removal operations that remove values inserted by~T are observed by 
all threads at the same time.
We investigate theoretical and practical properties of local linearizability
and its relationship to many existing consistency conditions.
We present a generic implementation method for 
locally linearizable data structures that uses existing linearizable data structures as building blocks. 
Our implementations show performance and scalability improvements over the 
original building blocks and outperform the fastest existing
container-type implementations.
\end{abstract}

\section{Introduction}

Concurrent data structures are pervasive all along the software stack, from
operating system code to application software and beyond.
Both correctness and performance are imperative for concurrent data
structure implementations.
Correctness is usually specified by relating concurrent executions, admitted by
the implementation, with sequential executions, admitted by the sequential
version of the data structure. The latter form the \emph{sequential specification} of the data structure.
This relationship is formally captured by {\em consistency conditions}, such as
linearizability, sequential consistency, or quiescent consistency~\cite{Herlihy:AMP08}.

Linearizability~\cite{HerlihyW90} is the most accepted consistency condition for concurrent data
structures due to its simplicity and general applicability.
It guarantees that the effects of all operations by all threads are observed consistently. 
This global visibility requirement imposes the need of extensive
synchronization among threads which may in turn jeopardize performance and
scalability.
In order to enhance performance and scalability of implementations, recent research has explored relaxed sequential specifications~\cite{Henzinger:POPL13,Shavit:CACM11,Afek:OPODIS10},  resulting in 
well-performing implementations of concurrent data structures~\cite{Afek:OPODIS10,Haas:CF13,Henzinger:POPL13,Kirsch:PACT13,Rihani:CORR14,Alistarh:PPOPP15}. 
Except for~\cite{Jagadeesan:ICALP14}, the space of alternative consistency conditions that relax linearizability has been left unexplored to a large extent.
In this paper, we explore (part of) this gap by investigating 
\emph{local linearizability}, a novel consistency condition that is applicable to a large 
class of concurrent data structures that we call \emph{container-type} data 
structures, or \emph{containers} for short.
Containers include pools, queues, and stacks. 
A fine-grained spectrum of consistency conditions enables us to 
describe the semantics of concurrent implementations more precisely, e.g.,
we show in our appendix that work stealing queues~\cite{Michael:PPOPP09} which could only be proven to be linearizable wrt pool are actually locally linearizable wrt double-ended queue.

\begin{wrapfigure}[10]{r}{0.38\textwidth}
	\centering
  \scalebox{0.75}{
	\begin{tikzpicture}
							
		\node (T1) at (0.25,.75) {$T_1$};
		\node (T2) at (0.25,0) {$T_2$};
							
		\draw[dotted,->] (0.75, .75) -- (6.75, .75);
		\draw[dotted,->] (0.75, 0) -- (6.75, 0);
		\draw[|-|] (1,.75) -- node[above] {$\enq{1}$} (2,.75);
		\draw[|-|] (4,.75) -- node[above] {$\deq{2}$} (5,.75);
		\draw[|-|] (2.5,0) -- node[above] {$\enq{2}$} (3.5,0);
		\draw[|-|] (5.5,0) -- node[above] {$\deq{1}$} (6.5,0);
							
		\draw[dashed, rounded corners=3pt] (1.5, 1.35) -- (0.9, 1.35) -- (0.9, .6) -- (3.5, .6) -- (5.5, -0.2) -- (6.6, -0.2)  -- (6.6, .6) -- (3.9, .6) -- (1.9, 1.35) -- (1.5, 1.35);
							
		\draw[rounded corners=3pt] (3, -0.2) -- (3.6, -0.2) -- (5.1, .55) -- (5.1, 1.35) -- (4.0, 1.35) -- (2.4, .5) -- (2.4, -0.2) -- (3, -0.2);
							
	\end{tikzpicture}}\\[3pt]
	\begin{tabularx}{15em}{@{}X@{}}
	\footnotesize The thread-induced history of thread~$T_1$ is 
enclosed by a dashed line while the thread-induced history of thread~$T_2$ is enclosed by a solid line.
	\end{tabularx}\\
	\caption{Local Linearizability}
	\label{fig:intro-ll-vs-lin}
\end{wrapfigure} 
Local linearizability is a (thread-)local consistency condition that 
guarantees that insertions \emph{per thread} are observed
consistently. 
While linearizability requires a consistent view over all insertions, we 
only require that projections of the global history---so called \emph{thread-induced histories}---are linearizable.
The induced history of a thread~$T$ is a projection of a program
execution to the insert-operations in $T$ combined with all remove-operations
that remove values inserted by $T$ irrespective of whether they happen in
$T$ or not.
Then, the program execution is locally linearizable iff each thread-induced 
history is linearizable.
Consider the example (sequential) history depicted in Figure~\ref{fig:intro-ll-vs-lin}.
It is not linearizable wrt a queue since the values are 
not dequeued in the same order as they were enqueued.
However, each thread-induced history is linearizable wrt a queue and, 
therefore, the overall execution is locally linearizable wrt a queue.
In contrast to semantic relaxations based on relaxing sequential semantics 
such as~\cite{Henzinger:POPL13,Afek:OPODIS10}, local linearizability coincides
with \emph{sequential correctness} for single-threaded histories, i.e.,
a single-threaded and, therefore, sequential history is locally linearizable 
wrt a given sequential specification if and only if it is admitted by the
sequential specification.

Local linearizability is to linearizability what coherence is to 
sequential consistency.
Coherence~\cite{Hennessy:Book11}, which is almost universally accepted as the 
absolute minimum that a shared memory system should satisfy, is the 
requirement that there exists a unique global order per shared memory 
location. 
Thus, while all accesses by all threads to a given memory location have to conform to a 
unique order, consistent with program order, the relative ordering of accesses 
to multiple memory locations do not have to be the same. 
In other words, coherence is sequential consistency per memory location. 
Similarly, local linearizability is linearizability per local history.
In our view, local linearizability offers enough 
consistency for the correctness of many applications as it is the local 
view of the client that often matters.
For example, in a locally linearizable queue each client (thread) has the 
impression of using a perfect queue---no reordering will ever be observed 
among the values inserted by a single thread. 
Such guarantees suffice for many e-commerce and cloud applications. 
Implementations of locally linearizable data structures have been successfully applied for managing free lists in the design of the fast and scalable memory allocator scalloc~\cite{Aigner:OOPSLA15}.
Moreover, except for fairness, locally linearizable queues guarantee all 
properties required from Dispatch Queues~\cite{dispatch-queues}, a common concurrency programming mechanism on mobile devices.

In this paper, we study theoretical and practical properties of local linearizability.
Local linearizability is compositional---a history over multiple concurrent 
objects is locally linearizable iff all per-object histories are 
locally linearizable (see Thm.~\ref{thm:comp}) and locally linearizable 
container-type data structures, including queues and stacks, admit only ``sane'' 
behaviours---no duplicated values, no values returned from thin air, and no 
values lost (see Prop.~\ref{prop:LL-pool}).
Local linearizability is a weakening of 
linearizability for a natural class of data structures including pools, 
queues, and stacks (see Sec.~\ref{sec:LL-L-comparison}). 
We compare local linearizability to  linearizability, sequential, and quiescent
consistency, and to many shared-memory consistency conditions.

Finally, local linearizability leads to new
efficient implementations.
We present a generic implementation scheme that, given a linearizable
implementation of a sequential specification~$S$, produces an implementation
that is locally linearizable wrt~$S$
(see Sec.~\ref{sec:implementations-queues-and-stacks}). 
Our implementations show dramatic improvements in performance and scalability. 
In most cases the locally linearizable implementations scale almost linearly and even outperform state-of-the-art pool implementations. 
We produced locally linearizable variants of state-of-the-art concurrent queues and stacks, as well as of the relaxed data structures from~\cite{Henzinger:POPL13,Kirsch:PACT13}. 
The latter are relaxed in two dimensions: they are locally linearizable (the consistency condition is relaxed) and are out-of-order-relaxed (the sequential specification is relaxed).
The speedup of the locally linearizable implementation to the
fastest linearizable queue (LCRQ) and stack (TS
Stack) implementation at 80 threads is 2.77 and 2.64, respectively. 
Verification of local linearizability, i.e. proving correctness, for each of our new locally linearizable implementations is immediate, given that the starting implementations are linearizable.

\section{Semantics of Concurrent Objects}
\label{sec:semantics}

The common approach to define the semantics of an implementation of a 
concurrent data structure is (1) to specify a set of valid sequential
behaviors---the sequential specification, and (2) to relate the admissible
concurrent executions to sequential executions specified by the sequential
specification---via the consistency condition.
That means that an implementation of a concurrent data structure actually
corresponds to several sequential data structures, and vice versa, depending 
on the consistency condition used. 
A (sequential) data structure $D$ is an object with a set of method calls~$\Sigma$. 
We assume that method calls include parameters, i.e., input and output values from a given set of values. 
The sequential specification $S$ of $D$ is a prefix-closed subset of $\Sigma^*$. 
The elements of $S$ are called $D$-valid sequences.
For ease of presentation, we assume that each value in a data structure can be inserted and removed at most once. 
This is without loss of generality, as we may see the set of values as consisting of pairs of elements (core values) and version numbers, i.e. $V = E \times \mathbb{N}$. Note that this is a technical assumption that only makes the presentation and the proofs simpler, it is not needed and not done in locally linearizable implementations. 
While elements may be inserted and removed multiple times, the version numbers provide uniqueness of values.
Our assumption ensures that whenever a sequence~\history{s} is part of a sequential specification $S$, then, each method call in \history{s} appears exactly once. 
An additional core value, that is not an element, is $\emptyValue$. 
It is returned by remove method calls that do not find an element to return.
We denote by $\Emp$ the set of values that are versions of $\emptyValue$, i.e., $\Emp = \{\emptyValue\} \times \mathbb{N}$.

\begin{definition}[Appears-before Order, Appears-in Relation]
Given a sequence $\history{s} \in \Sigma^*$ in which each method call appears
exactly once, we denote by $\abOrder{\history{s}}$ the total 
\emph{appears-before order} over method calls in $\history{s}$. 
Given a method call $m \in \Sigma$, we write $m \in \history{s}$ for \emph{$m$ appears in~$\history{s}$}. 
\qedd
\end{definition}

Throughout the paper, we will use pool, queue, and stack as typical examples 
of containers.
We specify their sequential specifications in an \emph{axiomatic}
way~\cite{Henzinger:CONCUR13}, i.e., as sets of axioms that exactly define the
valid sequences.

\begin{table}[t]
\centering
\footnotesize
\begin{tabularx}{\textwidth}{@{\hspace{0pt}}l@{\hspace{3pt}}X}\hline\\[-7pt]
(1) & $\forall i,j \in \{1,\dots, n\}. \,\,\,\history{s} = m_1\dots m_n \,\, \wedge \,\, m_i = m_j \,\,\, \Rightarrow \,\,\, i = j$\\[2pt]
(2) & $\forall x \in V.\,\,\, \rr{x} \in \history{s} \,\,\, \Rightarrow \,\,\, \ii{x} \in \history{s} \wedge \ii{x} \abOrder{\history{s}} \rr{x}$\\[2pt]
(3) & $\forall e \in \Emp. \,\,\forall x \in V.\,\,\, \ii{x} \abOrder{\history{s}}\rr{e} \Rightarrow \rr{x} \abOrder{\history{s}}\rr{e}$\\[2pt]
(4) & $\forall x,y \in V.\,\, \ii{x} \abOrder{\history{s}} \ii{y}\,\,\, \wedge\,\,\,\rr{y} \in \history{s} \,\,\,\Rightarrow\,\,\, \rr{x} \in \history{s} \,\, \wedge \,\, \rr{x} \abOrder{\history{s}} \rr{y}$\\[2pt]
(5) & $\forall x,y \in V.\,\, \ii{x} \abOrder{\history{s}} \ii{y} \abOrder{\history{s}} \rr{x} \,\,\,\Rightarrow\,\,\, 
\rr{y} \in \history{s} \,\, \wedge \,\, \rr{y} \abOrder{\history{s}} \rr{x}$\\[2pt]\hline
\end{tabularx}
\caption{The pool axioms (1), (2), (3); the queue order axiom (4); the stack order axiom (5)\label{tab:axioms}}
\end{table}

\begin{definition}[Pool, Queue, \& Stack]
\label{example:pool}
\label{example:queue} 
\label{example:stack}
A pool, queue, and stack with values in a set $V$ have the sets of 
methods 
$\Sigma_P = \{\ins{x}, \rem{x} \mid x \in V\} \cup \{\rem{e} \mid e \in \Emp\}$, 
$\Sigma_Q = \{\enq{x}, \deq{x} \mid x \in V\} \cup \{\deq{e} \mid e \in \Emp\}$,
and 
$\Sigma_S = \{\push{x},
\pop{x} \mid x \in V\} \cup \{\pop{e} \mid e \in \Emp\}$,
respectively.
We denote the sequential specification of a pool by $S_P$, the
sequential specification of a queue by $S_Q$, and the sequential
specification of a stack by $S_S$.
A sequence $\history{s} \in \Sigma_P^*$ belongs to $S_P$ iff 
it satisfies axioms (1) - (3) in Table~\ref{tab:axioms}---\emph{the pool axioms}---when instantiating $\ii$ with $\ins$ and $\rr$ with $\rem$. 
We keep axiom~(1) for completeness, although it is subsumed by our  assumption that each value is inserted and removed at most once.
Specification~$S_Q$ contains all sequences $\history{s}$ that satisfy the pool axioms and axiom (4)---\emph{the queue order axiom}---after instantiating $\ii$ with $\enq$ and $\rr$ with $\deq$. 
Finally, $S_S$ contains all sequences $\history{s}$ that satisfy 
the pool axioms and axiom (5)---\emph{the stack order axiom}---after instantiating $\ii$ with $\push$ and $\rr$ with $\pop$. 
\qedd
\end{definition}

We represent concurrent executions via concurrent histories. 
An example history is shown in Figure~\ref{fig:intro-ll-vs-lin}. 
Each thread executes a sequence of method calls from $\Sigma$; method calls executed by different threads may overlap (which does not happen in Figure~\ref{fig:intro-ll-vs-lin}). 
The real-time duration of method calls is irrelevant for the semantics of concurrent objects; all that matters is whether method calls overlap. Given this abstraction, a concurrent history is fully determined by a sequence of invocation and response events of method calls.
We distinguish method invocation and response events by augmenting the alphabet.
 Let $\Sigma_i = \{ m_i \mid m \in \Sigma\}$ and $\Sigma_r = \{ m_r \mid m \in \Sigma\}$ denote the sets of method-invocation events and method-response events, respectively, for the method calls in $\Sigma$. 
Moreover, let $I$ be the set of thread identifiers. 
 Let $\Sigma^I_i = \{ m^k_i \mid m \in \Sigma, k \in I\}$ and $\Sigma^I_r = \{ m^k_r \mid m \in \Sigma, k\in I\}$ denote the sets of method-invocation and -response events augmented with identifiers of executing threads. For example, $m^k_i$ is the invocation of method call $m$ by thread $k$. 
 Before we proceed, we mention a standard notion that we will need in several occasions.
 
\begin{definition}[Projection]
 \label{def:proj-seq}
 Let \history{s} be a sequence over alphabet $\Sigma$ and $M \subseteq \Sigma$. By $\history{s}|M$ we denote the projection of $\history{s}$ on the symbols in $M$, i.e., the sequence obtained from $\history{s}$ by removing all symbols that are not in $M$. \qedd
\end{definition}

\begin{definition}[History] 
A (concurrent) history $\history{h}$ is a sequence in $(\Sigma^I_i \cup \Sigma^I_r)^*$ where 
\begin{inparaenum}[\bf (1)]
\item no invocation or response event appears more than once, i.e., if $\history{h} = m_1 \dots m_n$ and $m_h = m_*^k(x)$ and $m_j = m_*^l(x)$, for $* \in \{i,r\}$, then $h = j$ and $k = l$, and
\item if a response event $m^k_r$ appears in $\history{h}$, then the corresponding invocation event $m^k_i$ also appears in $\history{h}$ and $m_i \abOrder{\history{h}} m_r$.\qedd
\end{inparaenum}
\end{definition}

\begin{example}\label{ex:queue-history}
A queue history (left) and its formal representation as a sequence (right):\\[2pt]
\begin{minipage}{0.425\textwidth}
	\centering
	\scalebox{0.75}{
	\begin{tikzpicture}
							
		\node (T1) at (0.25,.75) {$T_1$};
		\node (T2) at (0.25,0) {$T_2$};
							
		\draw[dotted,->] (0.75, .75) -- (6, .75);
		\draw[dotted,->] (0.75, 0) -- (6, 0);
		\draw[|-|] (1,.75) -- node[above] {$\enq{2}$} (2.5,.75);
		\draw[|-|] (4,.75) -- node[above] {$\deq{1}$} (5.5,.75);
		\draw[|-|] (2.25,0) -- node[above] {$\enq{1}$} (4.25,0);
							
	\end{tikzpicture}
	}
\end{minipage}
\hfill
\begin{minipage}{0.55\textwidth}
\scalebox{0.9}{
$ \enq{2}^1_i \enq{1}^2_i \enq{2}^1_r \deq{1}^1_i \enq{1}^2_r \deq{1}^1_r$
}
\end{minipage}
\end{example}

A history is \emph{sequential} if every response event is immediately
preceded by its matching invocation event and vice versa. 
Hence, we may ignore thread identifiers and identify a sequential history with a sequence in $\Sigma^*$,
e.g., $\enq{1}\enq{2}\deq{2}\deq{1}$ identifies the sequential history 
in Figure~\ref{fig:intro-ll-vs-lin}.

A history~$\history{h}$ is \emph{well-formed} if $\history{h}|k$ is sequential for every thread identifier $k \in I$ where $\history{h}|k$ denotes the projection of $\history{h}$ on the set $\{ m^k_i \mid m \in \Sigma\} \cup \{ m^k_r \mid m \in \Sigma\}$ of events that are local to thread~$k$. 
From now on we will use the term history for well-formed history. 
Also, we may omit thread identifiers if they are not essential in a discussion.

A history $\history{h}$ determines a partial order on its set of method calls, the precedence order:
\begin{definition}[Appears-in Relation, Precedence Order]
\label{def:prec-order} 
The set of method calls of a history~$\history{h}$ is $M({\history{h}}) = \{m  \mid m_i \in \history{h}\}$. 
 A method call $m$ appears in $\history{h}$, notation $m \in \history{h}$, if $m \in M({\history{h}})$.
The \emph{precedence order} for $\history{h}$ is the partial order $\precOrder{\history{h}}$ such that, for $m, n \in \history{h}$, 
we have that $m \precOrder{\history{h}} n$ iff $m_r \abOrder{\history{h}} n_i$.
By $<_{\history{h}}^k$ we denote $\precOrder{\history{h}|k}$, the subset of the precedence order that relates pairs of method calls of thread $k$, i.e., the program order of thread $k$. \qedd
\end{definition}

We can characterize a sequential history as a 
history whose precedence order is total. 
In particular, the precedence order $\precOrder{\history{s}}$ of a sequential history $\history{s}$ coincides with its appears-before order $\abOrder{\history{s}}$.
The total order for history~$\history{s}$ in Fig.~\ref{fig:intro-ll-vs-lin} is $\enq{1} \precOrder{\history{s}}\enq{2} \precOrder{\history{s}}\deq{2} \precOrder{\history{s}}\deq{1}$.

\begin{definition}[Projection to a set of method calls]
\label{def:projection} Let $\history{h}$ be a history, $M \subseteq \Sigma$, 
$M^I_i = \{m^k_i \mid m \in M, k\in I\}$, and $M^I_r = \{m^k_r \mid m \in M, k\in I\}$. Then, we write
 $\history{h}|M$ for $\history{h}| (M^I_i \cup M^I_r)$. \qedd
\end{definition}

\noindent
Note that $\history{h}|M$ inherits \history{h}'s precedence order:
$m \precOrder{\history{h}|M} n \,\,\,\Leftrightarrow\,\,\,  m \in M \,\,\wedge\,\, n\in M\,\,\wedge\,\,m \precOrder{\history{h}} n$

A history $\history{h}$ is \emph{complete} if the response of every invocation event in $\history{h}$ appears in $\history{h}$. 
Given a history $\history{h}$, $\Complete{\history{h}}$ denotes the set of all \emph{completions} of $\history{h}$, i.e., the set of all complete histories that are obtained from $\history{h}$ by appending missing response events and/or removing pending invocation events. 
Note that $\Complete{\history{h}} = \{\history{h}\}$ iff $\history{h}$ is a complete history. 

A concurrent data structure $D$ over a set of methods $\Sigma$ is a 
(prefix-closed) set of concurrent histories over $\Sigma$.
A history may involve several concurrent objects. 
Let $O$ be a set of concurrent objects with individual sets of method calls 
$\Sigma_q$ and sequential specifications~$S_q$ for each object $q \in O$. 
A history~$\history{h}$ over $O$ is a history over the (disjoint) union of 
method calls of all objects in $O$, i.e., it has a set of method calls 
$\bigcup_{q \in O} \{q.m \mid m \in \Sigma_q\}$. 
The added prefix $q.$ ensures that the union is disjoint. 
The \emph{projection} of $\history{h}$ to an object $q \in O$, denoted by $\history{h}|q$, is the history with a set of method calls $\Sigma_q$ obtained by removing the prefix~$q.$ in every method call in $\history{h}|\{q.m \mid m \in \Sigma_q\}$.

\begin{definition}[Linearizability~\cite{HerlihyW90}]
\label{def:linearizability}
A history $\history{h}$ is \emph{linearizable} wrt the sequential specification~$S$ if there is a sequential 
history~$\history{s} \in S$ and a completion $\history{h}_c \in \Complete{\history{h}}$ such that
\begin{inparaenum}[\bf (1)]
	\item $\history{s}$ is a permutation of $\history{h}_c$, 
	      and
	\item $\history{s}$ preserves the precedence order of $\history{h}_c$, i.e., if $m \precOrder{\history{h}_c} n$, then $m \precOrder{\history{s}} n$. 
\end{inparaenum}
We refer to $\history{s}$ as a \emph{linearization} of $\history{h}$.
A concurrent data structure $D$ is linearizable wrt~$S$
if every history~$\history{h}$ of~$D$ is linearizable wrt~$S$.
A history~$\history{h}$ over a set of concurrent objects $O$ is linearizable wrt the sequential specifications~$S_q$ for $q \in O$ if there exists a linearization $\history{s}$ of $\history{h}$ such that $\history{s}|q \in S_q$ for each object~$q \in O$. \qedd
\end{definition}

\section{Local Linearizability}
\label{sec:semantics:lc}

Local linearizability is applicable to containers whose set of method calls is
a disjoint union $\Sigma = \Ins \cup \Rem \cup \DOb \cup \SOb$
of insertion method calls $\Ins$, removal method calls $\Rem$, data-observation method calls $\DOb$, and (global) shape-observation method calls $\SOb$.
Insertions (removals) insert (remove) a \emph{single} value in the data set $V$ or \emptyValue; data observations return a \emph{single} value in $V$; shape observations return a value (not necessarily in $V$) that provides information on the shape of the state, for example,
the size of a data structure.
Examples of data observations are $\head{x}$ (queue), $\topp{x}$ (stack), and $\peek{x}$ (pool). 
Examples of shape observations are $\eempty{b}$ that returns true if the data structure is empty and false otherwise, and $\size{n}$ that returns the number of elements in the data structure. 
 
Even though we refrain from formal definitions, we want to stress that a valid sequence of a container remains valid after deleting observer method calls:
\begin{equation}
 	\label{eq:cont-obs} 
 	S \,| \left( \Ins \cup \Rem \right) \subseteq S .
\end{equation}

\noindent
There are also containers with multiple insert/remove methods, e.g., a double-ended queue (deque) is a container with insert-left, insert-right, remove-left, and remove-right methods, to which local linearizability is also applicable.
However, local linearizability requires 
that each method call is either an insertion, or a removal, or an observation. 
As a consequence, set is not a container according to our definition, as in a set $\ins{x}$ acts as a global observer first, checking whether (some version of) $x$ is already in the set, and if not inserts $x$. 
Also hash tables are not containers for a similar reason.  
 
Note that the arity of each method call in a container being one excludes data structures like snapshot objects. 
It is possible to deal with higher arities in a fairly natural way, however, at the cost of complicated presentation.
We chose to present local linearizability on simple containers only. 
\label{sec:LL-core-def}
We present the definition of local linearizability without shape observations here and discuss shape observations in Appendix~\ref{sec:LL-shape-obs}.

\begin{definition}[In- and out-methods]	\label{def:in-out-m}
Let~$\history{h}$ be a container history. 
For each thread $T$ we define two subsets of the methods in $\history{h}$, called in-methods~$\mathrm{I}_T$ and out-methods~$\mathrm{O}_T$ of thread~$T$, respectively: \\[2pt]
\begin{tabularx}{\textwidth}{@{}Xl@{\hspace{.5em}}c@{\hspace{.5em}}lX@{}}
	& $\mathrm{I}_T$ & = &  $\{m \mid m \in M(\history{h}|{T}) \cap \Ins\}$\\
	& $\mathrm{O}_T$ & = &  \,\,\,\, $\{m(a) \in M(\history{h}) \cap \Rem \mid \ins{a} \in \mathrm{I}_T\} \cup \{m(e) \in M(\history{h}) \cap \Rem \mid e \in \Emp \}$\\
	& & & $\cup~\{m(a) \in M(\history{h}) \cap \DOb \mid \ins{a} \in \mathrm{I}_T\}$. & \qedd
\end{tabularx}
\end{definition}

Hence, the in-methods for thread $T$ are all insertions performed by $T$. 
The out-methods are all removals and data observers that return values inserted by $T$. 
Removals that remove the value $\emptyValue$ are also automatically added to the out-methods of $T$ as any thread (and hence also $T$) could be the cause of ``inserting'' $\emptyValue$. This way, removals of $\emptyValue$ serve as means for global synchronization. Without them each thread could perform all its operations locally without ever communicating with the other threads.
Note that the out-methods~$\mathrm{O}_T$ of thread~$T$ need not be performed by $T$, but they return values that are inserted by $T$.

\begin{definition}[Thread-induced History]
Let $\history{h}$ be a history. The thread-induced history $\history{h}_T$ is the projection of $\history{h}$ to the in- and out-methods of thread~$T$, 
i.e., $\history{h}_T = \history{h}| 
\left(\mathrm{I}_T \cup \mathrm{O}_T\right)$. \qedd
\end{definition}

\begin{definition}[Local Linearizability]
\label{def:lc}
A history $\history{h}$ is locally linearizable wrt a sequential specification~$S$ if 
\begin{inparaenum}[\bf (1)]
\item each thread-induced history $\history{h}_T$ is linearizable wrt~$S$, and 
\item the thread-induced histories $\history{h}_T$ form a decomposition of $\history{h}$,  i.e., $m \in \history{h} \Rightarrow m \in \history{h}_T$ for some thread~$T$. 
\end{inparaenum}
A data structure~$D$ is locally linearizable wrt~$S$ if every history~$\history{h}$ of~$D$ 
is locally linearizable wrt~$S$.
A history~$\history{h}$ over a set of concurrent objects $O$ is locally linearizable wrt the sequential specifications~$S_q$ for $q \in O$ if each thread-induced history is linearizable over~$O$ and the thread-induced histories form a decomposition of $\history{h}$, i.e., $q.m \in \history{h} \Rightarrow q.m \in \history{h}_T$ for some thread~$T$. \qedd
\end{definition}

Local linearizability is sequentially correct, i.e., a single-threaded (necessarily sequential) history $\history{h}$ is locally linearizable wrt a sequential specification $S$ iff $\history{h} \in S$.
Like linearizability~\cite{Herlihy:AMP08}, 
local linearizability is compositional. 
The complete proof of the following theorem and missing or extended proofs of all following properties can be found in Appendix~\ref{sec:proofs}.

\begin{theorem}[Compositionality]
\label{thm:comp}
A history $\history{h}$ over a set of objects~$O$ with sequential 
specifications~$S_q$ for $q \in O$ is locally linearizable iff 
$\history{h}|q$ is locally linearizable wrt~$S_q$ for 
every~$q \in O$.
\end{theorem}

\begin{proof}[Proof (Sketch)]
The property follows from the compositionality of linearizability and the fact 
that $(\history{h}|q)_T = \history{h}_T|q$ for every thread~$T$ and 
object~$q$. 
\end{proof}

\subparagraph*{The Choices Made.}
Splitting a global history into subhistories and requiring consistency 
for each of them is central to local linearizability.
While this is common in shared-memory consistency
conditions~\cite{Hennessy:Book11,Lipton:TR88,Lipton:patent93,Ahamad:SPAA93,Goodman:TR91,Ahamad:DC95,Heddaya:TR92}, 
our study of local linearizability is a first step in exploring
subhistory-based consistency conditions for concurrent objects. 

We chose thread-induced subhistories since thread-locality reduces contention 
in concurrent objects and is known to lead to high performance as 
confirmed by our experiments.
To assign method calls to thread-induced histories, we took a data-centric 
point of view by (1) associating data values to threads, and (2) gathering all
method calls that insert/return a data value into the subhistory of the
associated thread  (Def.~\ref{def:in-out-m}). 
We associate data values to the thread that inserts them.
One can think of alternative approaches, for example, associate with a 
thread the values that it removed.
In our view, the advantages of our choice are clear:
First, by assigning inserted values to threads, every value in the history is assigned to some thread. 
In contrast, in the alternative approach, it is not clear where to assign the
values that are inserted but not removed.
Second, assigning inserted values to the inserting thread enables eager removals and ensures progress in locally linearizable data structures. 
In the alternative approach, it seems like the semantics of removing 
$\emptyValue$ should be local.

An orthogonal issue is to assign values from shape observations to
threads.
In Appendix~\ref{sec:LL-shape-obs}, we discuss two meaningful approaches and
show how local linearizability can be extended towards 
shape and data observations that appear in insertion operations of sets.

Finally, we have to choose a consistency condition required for each of the 
subhistories.
We chose linearizability as it is the best (strong) consistency condition for concurrent objects.

\section{Local Linearizability vs. Linearizability}
\label{sec:LL-L-comparison}

We now investigate the connection between local linearizability and linearizability.

\begin{proposition}[Lin 1]
\label{prop:lin-loc1}
In general, linearizability does not imply local linearizability.
\end{proposition}

\begin{proof}
We provide an example of a data structure that is linearizable but not locally linearizable.
Consider a sequential specification~$S_\crazy$ which behaves like a 
queue except when the first two insertions were 
performed without a removal in between---then the first two elements 
are removed out of order.
Formally, $\history{s} \in S_\crazy$ iff  
\begin{inparaenum}[\bf (1)]
\item $\history{s}\:= \history{s_1}\enq{a}\enq{b}\history{s_2}\deq{b}\history{s_3}\deq{a}\history{s_4}$ where 
$\history{s_1}\enq{a}\enq{b}\history{s_2}\deq{a}\history{s_3}\deq{b}\history{s_4} \in S_Q$ and $\history{s_1} \in \{\deq{e} \mid e \in \Emp\}^*$ for some $a, b \in V$, or
\item $\history{s} \in S_Q$ and $\history{s} \neq \history{s_1}\enq{a}\enq{b}\history{s_2}$ for $\history{s_1} \in \{\deq{e} \mid e \in \Emp\}^*$ and $a, b \in V$.
\end{inparaenum}
The example below is linearizable wrt~$S_\crazy$.
However, $T_1$'s induced history
$\enq{1}\enq{2}\deq{1}\deq{2}$ is not.\\[7pt]
\noindent
\begin{minipage}[t]{\textwidth}
\centering
\scalebox{0.75}{
\begin{tikzpicture}

	\node (T1) at (0.25,.75) {$T_1$};
	\node (T2) at (0.25,0) {$T_2$};

	\draw[dotted,->] (0.75, .75) -- (8.35, .75);
	\draw[dotted,->] (0.75, 0) -- (8.35, 0);
	\draw[|-|] (.85, .75) -- node[above] {$\enq{1}$} (1.85, .75);
	\draw[|-|] (3.35, .75) -- node[above] {$\enq{2}$} (4.35, .75);
	\draw[|-|] (4.6, .75) -- node[above] {$\deq{3}$} (5.6, .75);
	\draw[|-|] (7.1, .75) -- node[above] {$\deq{2}$} (8.1, .75);
	
	\draw[|-|] (2.1, 0) -- node[above] {$\enq{3}$} (3.1, 0);
	\draw[|-|] (5.85, 0) -- node[above] {$\deq{1}$} (6.85, 0);

\end{tikzpicture}
}
\end{minipage}\\[-1em]
\end{proof}

The following condition on a data structure specification is sufficient for
linearizability to imply local linearizability and is satisfied, e.g., by 
pool, queue, and stack. 

\begin{definition}[Closure under Data-Projection]
\label{def:data-projection-closedness}
A seq.~specification~$S$ over $\Sigma$ is \emph{closed under data-projection}\footnote{The same notion has been used in~\cite{Bouajjani:ICALP15} under the name \emph{closure under projection}.}
iff for all $\history{s}\in S$ and all $V' \subseteq V$, $\history{s}|\{ m(x) \in \Sigma \mid x \in V' \cup \Emp\}\in S$.\qedd
\end{definition}

\noindent
For 
$\history{s} = \enq{1}\enq{3}\enq{2}\deq{3}\deq{1}\deq{2}$ we have $\history{s}\in S_\crazy$, 
but $\history{s}|\{ \enq{x}, \deq{x} \mid x \in \{1,2\} \cup \Emp\}\notin S_\crazy$, i.e., 
$S_\crazy$ is not closed under data-projection.

\begin{proposition}[Lin 2]
	\label{prop:lin-loc2}
Linearizability implies local linearizability for sequential specifications that
are closed under data-projection.
\end{proposition}

\begin{proof}[Proof (Sketch)]
The property follows from Definition~\ref{def:data-projection-closedness} and Equation~(\ref{eq:cont-obs}).
\end{proof}

\noindent
There exist corner cases where local linearizability coincides with linearizability, e.g., for $S = \emptyset$ or $S = \Sigma^*$, or for single-producer/multiple-consumer histories.

We now turn our attention to pool, queue, and stack.

\begin{proposition}
\label{prop:DS-are-closed-under-projection}
The seq. specifications $S_P$, $S_Q$, and $S_S$ are closed under data-projection.
\end{proposition}

\begin{proof}[Proof (Sketch)]
Let $\history{s} \in S_P$, $V' \subseteq V$, and let $\history{s'} = \history{s}|\left(\{ \ins{x}, \rem{x} \mid x \in V' \cup\Emp\}\right)$.
Then, it suffices to check that all axioms for pool (Definition~\ref{example:pool} and Table~\ref{tab:axioms}) hold for $\history{s'}$.
\end{proof}

\begin{theorem}[Pool \& Queue \& Stack, Lin]
\label{thm:queues-and-stacks}
For pool, queue, and stack, local linearizability is (strictly) weaker than linearizability. 
\end{theorem}

\begin{proof}
Linearizability implies local linearizability for pool, queue, and stack as a consequence of
Proposition~\ref{prop:lin-loc2} and Proposition~\ref{prop:DS-are-closed-under-projection}.
The history in Figure~\ref{fig:ll-notsc} is locally linearizable but not linearizable wrt pool, 
queue and stack (after suitable renaming of method calls).
\end{proof}

Although local linearizability wrt a pool does not imply linearizability wrt a pool (Theorem~\ref{thm:queues-and-stacks}), it still guarantees several properties that ensure sane behavior as stated next.

\begin{proposition}[LocLin Pool]
\label{prop:LL-pool}
Let $\history{h}$ be a locally linearizable history wrt a pool. Then:\\[2pt]
\begin{tabularx}{\textwidth}{@{}l@{\hspace{5pt}}X@{}}
\normalsize 1. & \normalsize No value is duplicated, i.e., every remove method appears in $\history{h}$ at most once.\\[2pt]
\normalsize 2. & \normalsize No out-of-thin-air values, i.e., $\forall x \in V. \,\, \rem{x} \in \history{h} \Rightarrow \ins{x} \in \history{h} \,\,\,\wedge\,\,\, \rem{x} {\not<}_{\history{h}} \ins{x}.$\\[2pt]
\normalsize 3. & \normalsize No value is lost, i.e., $\forall x \in V. \,\,\forall e \in \Emp.\,\,\rem{e} \precOrder{\history{h}} \rem{x} \Rightarrow \ins{x} \not<_{\history{h}} \rem{e}$ and\\[2pt]
& \normalsize $\forall x \in V. \,\, \forall e \in \Emp. \,\, \ins{x} \precOrder{\history{h}} \rem{e}
\Rightarrow \rem{x} \in \history{h} \,\,\,\wedge\,\,\, \rem{e} {\not<}_{\history{h}} \rem{x}$.\\
\end{tabularx}
\end{proposition}

\begin{proof}
By direct unfolding of the definitions.
\end{proof}

Note that if a history $\history{h}$ is linearizable wrt a pool, then all of the three stated properties hold, as a consequence of linearizability and the definition of $S_P$.

\section{Local Linearizability vs. Other Relaxed Consistency Conditions}
\label{sec:other-consistency-conditions}

We compare local linearizability with other classical consistency conditions to better understand its guarantees and implications.

\begin{figure}[tb]
\centering
\begin{minipage}{.475\textwidth}
	\centering
	\scalebox{0.75}{
		\begin{tikzpicture}
										
			\node (T1) at (0.25,.75) {$T_1$};
			\node (T2) at (0.25,0) {$T_2$};
									
			\draw[dotted,->] (0.75, .75) -- (5.5, .75);
			\draw[dotted,->] (0.75, 0) -- (5.5, 0);
			\draw[|-|] (1,.75) -- node[above] {$\ii{1}$} (1.75,.75);
			\draw[|-|] (2,.75) -- node[above] {$\rr{\emptyValue}$} (4.25,.75);
      \draw[|-|] (2.25,0) -- node[above] {$\ii{2}$} (3,0);
			\draw[|-|] (3.25,0) -- node[above] {$\rr{1}$} (4,0);
			\draw[|-|] (4.5,.75) -- node[above] {$\rr{2}$} (5.25,.75);
										
		\end{tikzpicture}
	}
	\caption{LL, not SC (Pool, Queue, Stack)}
	\label{fig:ll-notsc}
\end{minipage}
\hfill
\begin{minipage}{.475\textwidth}
	\centering
	\scalebox{0.75}{
		\begin{tikzpicture}
		
			\node (T1) at (0.25,.75) {$T_1$};
			\node (T2) at (0.25,0) {$T_2$};
		
			\draw[dotted,->] (0.75, .75) -- (5, .75);
			\draw[dotted,->] (0.75, 0) -- (5, 0);
			\draw[|-|] (1,.75) -- node[above] {$\ii{1}$} (1.75,.75);
			\draw[|-|] (4,.75) -- node[above] {$\rr{1}$} (4.75,.75);
			\draw[|-|] (2,0) -- node[above] {$\rr{\emptyValue}$} (3.75,0);
			
		\end{tikzpicture}
	}
	\caption{SC, not LL (Pool, Queue, Stack)}
	\label{fig:sc-notll}	
\end{minipage}
\end{figure}

\subparagraph*{Sequential Consistency (SC).}
\label{sec:LL-SC-comparison}
A history~$\history{h}$ is \emph{sequentially consistent}~\cite{Herlihy:AMP08,Lamport:TC79} wrt a sequential specification $S$, if there exists a sequential history~$\history{s} \in S$ and a completion $\history{h}_c \in \Complete{\history{h}}$ such that 
\begin{inparaenum}[(1)]
	\item $\history{s}$ is a permutation of $\history{h}_c$, 
	      and
	\item $\history{s}$ preserves each thread's program order, i.e., if $m \precOrder{\history{h}}^T n$, for some thread $T$, then $m \precOrder{\history{s}} n$.
\end{inparaenum}
We refer to $\history{s}$ as a \emph{sequential witness} of $\history{h}$.
A data structure~$D$ is sequentially consistent wrt~$S$
if every history~$\history{h}$ of $D$ is sequentially consistent wrt~$S$.

Sequential consistency is a useful consistency condition for shared memory but it is not really suitable for data structures as it allows for behavior that excludes any coordination between threads~\cite{Sezgin:CORR15}: an implementation of a data structure in which every thread uses a dedicated copy of a sequential data structure without any synchronization is sequentially consistent. 
A sequentially consistent queue might always return \emptyValue in one (consumer) thread as the point in 
time of the operation can be moved, e.g., see Figure~\ref{fig:sc-notll}.
In a producer-consumer scenario such a queue might end up
with some threads not doing any work.

\begin{theorem}[Pool, Queue \& Stack, SC]
\label{thm:sc:queues-and-stacks}
For pool, queue, and stack, local linearizability is incomparable to sequential consistency.\qed
\end{theorem}

\noindent 
Figures~\ref{fig:ll-notsc} and~\ref{fig:sc-notll} give example histories that 
show the statement of Theorem~\ref{thm:sc:queues-and-stacks}.
In contrast to local linearizability, sequential consistency is not
compositional~\cite{Herlihy:AMP08}.

\subparagraph*{(Quantitative) Quiescent Consistency (QC \& QQC).}
\label{sec:quiesc-cons}
Like linearizability and sequential consistency, quiescent
consistency~\cite{Derrick:FM14,Herlihy:AMP08} also requires the existence of a 
sequential history, a \emph{quiescent witness}, that satisfies the sequential specification.
All three consistency conditions impose an order on the method calls of a 
concurrent history that a witness has to preserve.
Quiescent consistency uses the concept of \emph{quiescent states} to relax the requirement of preserving the precedence order imposed by linearizability.
A quiescent state is a point in a history at which there are no pending invocation events (all invoked method calls have already responded). 
In a quiescent witness, a method call~$m$ has to appear before a method 
call~$n$ if and only if there is a quiescent state between $m$ and $n$.
Method calls between two consecutive quiescent states can be ordered arbitrarily.
\emph{Quantitative quiescent consistency}~\cite{Jagadeesan:ICALP14} refines quiescent consistency by bounding the number of reorderings of operations 
between two quiescent states based on the concurrent behavior between these
two states.

The next result about quiescent consistency for pool is needed to establish the connection between quiescent consistency and local linearizability.

\begin{proposition}
\label{prop:pool-qc}
A pool history~$\history{h}$ satisfying 1.-3.~of Prop.~\ref{prop:LL-pool} is quiescently consistent. \qed
\end{proposition}

\noindent From Prop.~\ref{prop:LL-pool}~and~\ref{prop:pool-qc} follows
that local linearizability implies quiescent consistency for pool. 

\begin{theorem}[Pool, Queue \& Stack, QC]
\label{thm:qc:pool-queue-stack}
For pool, local linearizability is (strictly) stronger than quiescent consistency. For queue and stack, local linearizability is incomparable to quiescent consistency.\qed
\end{theorem}

Local linearizability also does not imply the stronger condition of quantitative quiescent consistency. 
Like local linearizability, quiescent consistency and quantitative quiescent
consistency are compositional~\cite{Herlihy:AMP08,Jagadeesan:ICALP14}.
For details, please see Appendix~\ref{appendix:qqc}.

\subparagraph*{Consistency Conditions for Distributed Shared Memory.}
There is extensive research on consistency conditions for distributed shared 
memory~\cite{Ahamad:SPAA93,Ahamad:DC95,Burckhardt:POPL14,Goodman:TR91,Heddaya:TR92,Hennessy:Book11,Lamport:TC79,Lipton:TR88,Lipton:patent93}.
In Appendix~\ref{appendix:ccdsm}, we compare local linearizability against coherence, PRAM consistency, 
processor consistency, causal consistency, and local consistency.
All these conditions split a history into subhistories and 
require consistency of the subhistories.
For our comparison, we first define a sequential specification~$S_M$ for a 
single memory location.
We assume that each memory location is preinitialized with a 
value~$v_\mathit{init} \in V$.
A read-operation returns the value of the last write-operation that was 
performed on the memory location or $v_\mathit{init}$ if there was no
write-operation.
We denote write-operations by \insName and read-operations by \headName.
Formally, we define $S_M$ as $S_M = \{ \head{v_\mathit{init}} \}^\star \cdot \{ \ins{v}\head{v}^i \mid i \geq 0, v \in V \}^\star$.
Note that read-operations are data observations and the same value can be read multiple times.
For brevity, we only consider histories that involve a 
single memory location. 
In the following, we summarize our comparison. For details, please see 
Appendix~\ref{appendix:ccdsm}.

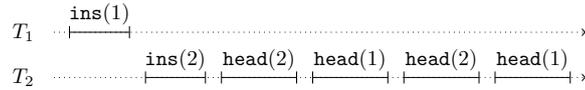
\begin{figure}[tb]
\centering
\scalebox{0.8}{
	  \begin{tikzpicture}
			
		  \node (T1) at (0.25,.75) {$T_1$};
  		\node (T2) at (0.25,0) {$T_2$};

	  	\draw[dotted,->] (0.75, .75) -- (9.5, .75);
  		\draw[dotted,->] (0.75, 0) -- (9.5, 0);
		  \draw[|-|] (1,.75) -- node[above] {$\ins{1}$} (2,.75);
		  \draw[|-|] (2.25,0) -- node[above] {$\ins{2}$} (3.25,0);
		  \draw[|-|] (3.5,0) -- node[above] {$\head{2}$} (4.75,0);
		  \draw[|-|] (5,0) -- node[above] {$\head{1}$} (6.25,0);
		  \draw[|-|] (6.5,0) -- node[above] {$\head{2}$} (7.75,0);
		  \draw[|-|] (8,0) -- node[above] {$\head{1}$} (9.25,0);
			
	  \end{tikzpicture}
	}
	\caption{Problematic shared-memory history.}
	\label{fig:problematic-dsm-history}
\end{figure}
While local linearizability is well-suited for 
concurrent data structures, this is not 
necessarily true for the mentioned shared-memory consistency conditions.
On the other hand, 
local linearizability appears to be problematic for shared memory.
Consider the locally linearizable history in Figure~\ref{fig:problematic-dsm-history}.
There, the read values oscillate between different values that were written 
by different threads.
Therefore, local linearizability does not imply any of the shared-memory consistency conditions.
In Appendix~\ref{appendix:ccdsm}, we further show that local linearizability is incomparable to all considered shared-memory conditions.

\lstset{
  numbers=left, %
  numberstyle=\tiny, %
  numbersep=4pt, %
  language=C, %
  commentstyle=\scriptsize\ttfamily, %
  fontadjust=true, %
  float=*, %
  basicstyle=\scriptsize\ttfamily, %
  keywordstyle=\scriptsize\ttfamily\bf, %
  tabsize=2, %
  firstnumber=1, %
  escapeinside={(*}{*)}, %
  breaklines=true,
}

\section{Locally Linearizable Implementations}
\label{sec:implementations-queues-and-stacks}

In this section, we focus on locally linearizable data structure 
implementations that are generic as follows:
Choose a linearizable implementation of a data structure $\Phi$ wrt a 
sequential specification $S_\Phi$, and we turn it into a (distributed) data structure called LLD $\Phi$ that is locally linearizable wrt $S_\Phi$. 
An LLD implementation takes several copies of $\Phi$ (that we call backends) and assigns to each
thread $T$ a backend $\Phi_T$.
Then, when thread $T$ inserts an element into LLD $\Phi$, the element is 
inserted into $\Phi_T$, and when an arbitrary thread removes an element from 
LLD $\Phi$, the element is removed from some $\Phi_T$ eagerly, i.e., if no 
element is found in the attempted backend $\Phi_{T}$ the search for an element
continues through all other backends. 
If no element is found in one round through the backends, then we return \emptyValue.

\begin{proposition}[LLD correctness]
\label{prop:LLD}
Let $\Phi$ be a data structure implementation that is linearizable wrt a sequential specification $S_\Phi$. Then LLD $\Phi$ is locally linearizable wrt $S_\Phi$. 
\end{proposition}

\begin{proof} 
Let $\history{h}$ be a history of LLD $\Phi$. 
The crucial observation is that each thread-induced history~$\history{h_T}$ is 
a backend history of $\Phi_T$ and hence linearizable wrt $S_\Phi$.
\end{proof}

Any number of copies (backends) is allowed in this generic implementation of LLD $\Phi$. 
If we take just one copy, we end up with a linearizable implementation. 
Also, any way of choosing a backend for removals is fine. 
However, both the number of backends and the backend selection strategy upon removals affect the performance significantly. 
In our LLD $\Phi$ implementations we use one backend per thread, resulting in no contention on insertions, and always attempt a local remove first. 
If this does not return an element, then we continue a search through all other backends starting from a randomly chosen backend.

LLD $\Phi$ is an implementation closely related to Distributed Queues~(DQs)~\cite{Haas:CF13}. 
A DQ is a (linearizable) \emph{pool} that is organized as a single segment 
of length $\ell$ holding $\ell$ backends.
DQs come in different flavours depending on how insert and remove methods are
distributed across the segment when accessing backends. 
No DQ variant in~\cite{Haas:CF13} follows the LLD approach described above. Moreover, while DQ algorithms are 
implemented for a fixed number of backends, LLD $\Phi$ implementations manage a segment of variable size, one backend per (active) thread.
Note that the strategy of selecting backends in the LLD $\Phi$ implementations is similar to other work in work stealing~\cite{Michael:PPOPP09}. 
However, in contrast to this work our data structures neither duplicate nor
lose elements.
LLD (stack) implementations have been successfully applied  for managing free lists in the fast and scalable memory allocator scalloc~\cite{Aigner:OOPSLA15}. The guarantees provided by local linearizability are not needed for the correctness of scalloc, i.e., the free lists could also use a weak pool (pool without a linearizable emptiness check). However, the LLD stack implementations provide good caching behavior when threads operate on their local stacks whereas a weak pool would potentially negatively impact performance.

We have implemented LLD variants of strict and relaxed queue and stack implementations. 
None of our implementations involves observation methods, but the LLD
algorithm can easily be extended to support observation methods.
For details, please see App.~\ref{sec:lld_and_observers}.
Finally, let us note that we have also experimented with other locally
linearizable implementations that 
lacked the genericity of the LLD implementations, and 
whose performance evaluation did not show promising results (see App.~\ref{sec:implementations}). 
As shown in Sec.~\ref{sec:LL-L-comparison}, a locally linearizable 
pool is not a linearizable pool, i.e., it lacks a linearizable
emptiness check. 
Indeed, LLD implementations do not provide a linearizable emptiness check, despite of eager removes. 
We provide LL$^+$D~$\Phi$, a variant of LLD~$\Phi$, that provides a linearizable emptiness check under mild conditions on the starting implementation~$\Phi$ (see App.~\ref{sec:llplusd} for details).

\subparagraph*{Experimental Evaluation.}\label{sec:experiments}
All experiments ran on a uniform memory architecture (UMA) machine with four
$10$-core 2GHz Intel Xeon \mbox{E7-4850} processors supporting two hardware
threads (hyperthreads) per core, 128GB of main memory, and Linux kernel version 3.8.0.
We also ran the experiments without hyper-threading resulting in no noticeable difference.
The CPU governor has been disabled.  
All measurements were obtained from the artifact-evaluated
Scal~benchmarking~framework~\cite{scal-framework,HaasHKLPS15,scal-artifact}, where you can also find the code of all
involved data structures. 
Scal uses preallocated memory (without freeing it) to avoid
memory management artifacts. 
For all measurements we report the arithmetic mean and the $95\%$ confidence
interval (sample size=10, corrected sample standard deviation).

In our experiments, we consider the linearizable queues 
Michael-Scott~queue (MS)~\cite{Michael:PODC96} and 
LCRQ~\cite{Morrison:PPOPP13} (improved version~\cite{multicore-computing-group-lcrq}),
the linearizable stacks 
Treiber stack (Treiber)~\cite{Treiber86} and 
TS stack~\cite{Dodds:POPL15},
the $k$-out-of-order relaxed 
$k$-FIFO queue~\cite{Kirsch:PACT13} and
$k$-Stack~\cite{Henzinger:POPL13}
and
linearizable well-performing pools based on distributed queues using random
balancing~\cite{Haas:CF13} (1-RA DQ for queue, and 1-RA DS for stack).
For each of these implementations (but the pools) we provide LLD variants (LLD LCRQ, LLD TS stack, LLD $k$-FIFO, and LLD $k$-Stack) and, when possible, LL$^+$D variants (LL$^+$D MS queue and LL$^+$D Treiber stack). 
Making the pools
locally linearizable is not promising as they are already distributed. 
Whenever LL$^+$D is achievable for a data structure implementation~$\Phi$ we
present 
only results for LL$^+$D~$\Phi$ 
as, in our workloads, LLD~$\Phi$ and LL$^+$D~$\Phi$ implementations 
perform with no visible difference.

We evaluate the data structures on a Scal producer-consumer benchmark where each producer and consumer is
configured to execute $10^6$ operations. 
To control contention, we add a busy wait of $5\mu s$ between operations. This is important as too high contention results in measuring hardware or operating system (e.g., scheduling) artifacts.
The number of threads ranges between $2$ and
$80$ (number of hardware threads) half of which are producers and half consumers.
To relate performance and
scalability we report the number of data structure operations per second.
Data structures that require parameters to be set are configured to allow
maximum parallelism for the producer-consumer workload with 80 threads. This results in
$k=80$ for all $k$-FIFO and $k$-Stack variants (40~producers and 40~consumers in
parallel on a single segment), $p=80$ for 1-RA-DQ and 1-RA-DS (40~producers and 40~consumers in parallel on
different backends). The TS Stack algorithm also needs to be configured with a delay parameter. We use optimal delay ($7\mu s$) for the TS Stack and zero delay for the LLD TS Stack,
as delays degrade the performance of the LLD implementation.

\begin{figure}[t!]
\centering
\begin{minipage}{\textwidth}
	\centering
	\resizebox{.75\textwidth}{!}{%
		\begingroup
		  \makeatletter
		  \gdef\gplbacktext{}%
		  \gdef\gplfronttext{}%
		  \makeatother
		      \def\colorrgb#1{\color[rgb]{#1}}%
		      \def\colorgray#1{\color[gray]{#1}}%
		      \expandafter\def\csname LTw\endcsname{\color{white}}%
		      \expandafter\def\csname LTb\endcsname{\color{black}}%
		      \expandafter\def\csname LTa\endcsname{\color{black}}%
		      \expandafter\def\csname LT0\endcsname{\color[rgb]{1,0,0}}%
		      \expandafter\def\csname LT1\endcsname{\color[rgb]{0,1,0}}%
		      \expandafter\def\csname LT2\endcsname{\color[rgb]{0,0,1}}%
		      \expandafter\def\csname LT3\endcsname{\color[rgb]{1,0,1}}%
		      \expandafter\def\csname LT4\endcsname{\color[rgb]{0,1,1}}%
		      \expandafter\def\csname LT5\endcsname{\color[rgb]{1,1,0}}%
		      \expandafter\def\csname LT6\endcsname{\color[rgb]{0,0,0}}%
		      \expandafter\def\csname LT7\endcsname{\color[rgb]{1,0.3,0}}%
		      \expandafter\def\csname LT8\endcsname{\color[rgb]{0.5,0.5,0.5}}%
		  \setlength{\unitlength}{0.0500bp}%
		  \begin{picture}(7200.00,4000.00)(1100.00, 900.00)%
		    \gplbacktext{%
		      \csname LTb\endcsname%
		      \put(814,1584){\makebox(0,0)[r]{\strut{} 0}}%
		      \put(814,1829){\makebox(0,0)[r]{\strut{} 2}}%
		      \put(814,2075){\makebox(0,0)[r]{\strut{} 4}}%
		      \put(814,2320){\makebox(0,0)[r]{\strut{} 6}}%
		      \put(814,2566){\makebox(0,0)[r]{\strut{} 8}}%
		      \put(814,2811){\makebox(0,0)[r]{\strut{} 10}}%
		      \put(814,3057){\makebox(0,0)[r]{\strut{} 12}}%
		      \put(814,3302){\makebox(0,0)[r]{\strut{} 14}}%
		      \put(814,3548){\makebox(0,0)[r]{\strut{} 16}}%
		      \put(814,3793){\makebox(0,0)[r]{\strut{} 18}}%
		      \put(814,4039){\makebox(0,0)[r]{\strut{} 20}}%
		      \put(814,4284){\makebox(0,0)[r]{\strut{} 22}}%
		      \put(814,4530){\makebox(0,0)[r]{\strut{} 24}}%
		      \put(814,4775){\makebox(0,0)[r]{\strut{} 26}}%
		      \put(1019,1364){\makebox(0,0){\strut{}2}}%
		      \put(1605,1364){\makebox(0,0){\strut{}10}}%
		      \put(2337,1364){\makebox(0,0){\strut{}20}}%
		      \put(3069,1364){\makebox(0,0){\strut{}30}}%
		      \put(3801,1364){\makebox(0,0){\strut{}40}}%
		      \put(4533,1364){\makebox(0,0){\strut{}50}}%
		      \put(5266,1364){\makebox(0,0){\strut{}60}}%
		      \put(5998,1364){\makebox(0,0){\strut{}70}}%
		      \put(6730,1364){\makebox(0,0){\strut{}80}}%
		      \put(176,3179){\rotatebox{-270}{\makebox(0,0){\strut{}million operations per sec (more is better)}}}%
		      \put(3874,1034){\makebox(0,0){\strut{}number of threads}}%
		    }%
		    \gplfronttext{%
		      \csname LTb\endcsname%
		      \put(8500,3880){\makebox(0,0)[r]{\strut{}MS}}%
		      \csname LTb\endcsname%
		      \put(8500,3660){\makebox(0,0)[r]{\strut{}LCRQ}}%
		      \csname LTb\endcsname%
		      \put(8500,3440){\makebox(0,0)[r]{\strut{}k-FIFO}}%
		      \csname LTb\endcsname%
		      \put(8500,3220){\makebox(0,0)[r]{\strut{}LL+D~MS}}%
		      \csname LTb\endcsname%
		      \put(8500,3000){\makebox(0,0)[r]{\strut{}LLD~LCRQ}}%
		      \csname LTb\endcsname%
		      \put(8500,2780){\makebox(0,0)[r]{\strut{}LLD~k-FIFO}}%
		      \csname LTb\endcsname%
		      \put(8500,2560){\makebox(0,0)[r]{\strut{}1-RA~DQ}}%
		    }%
		    \gplbacktext
		    \put(0,1000){\includegraphics[keepaspectratio=true,trim=0 50 0 0,clip]{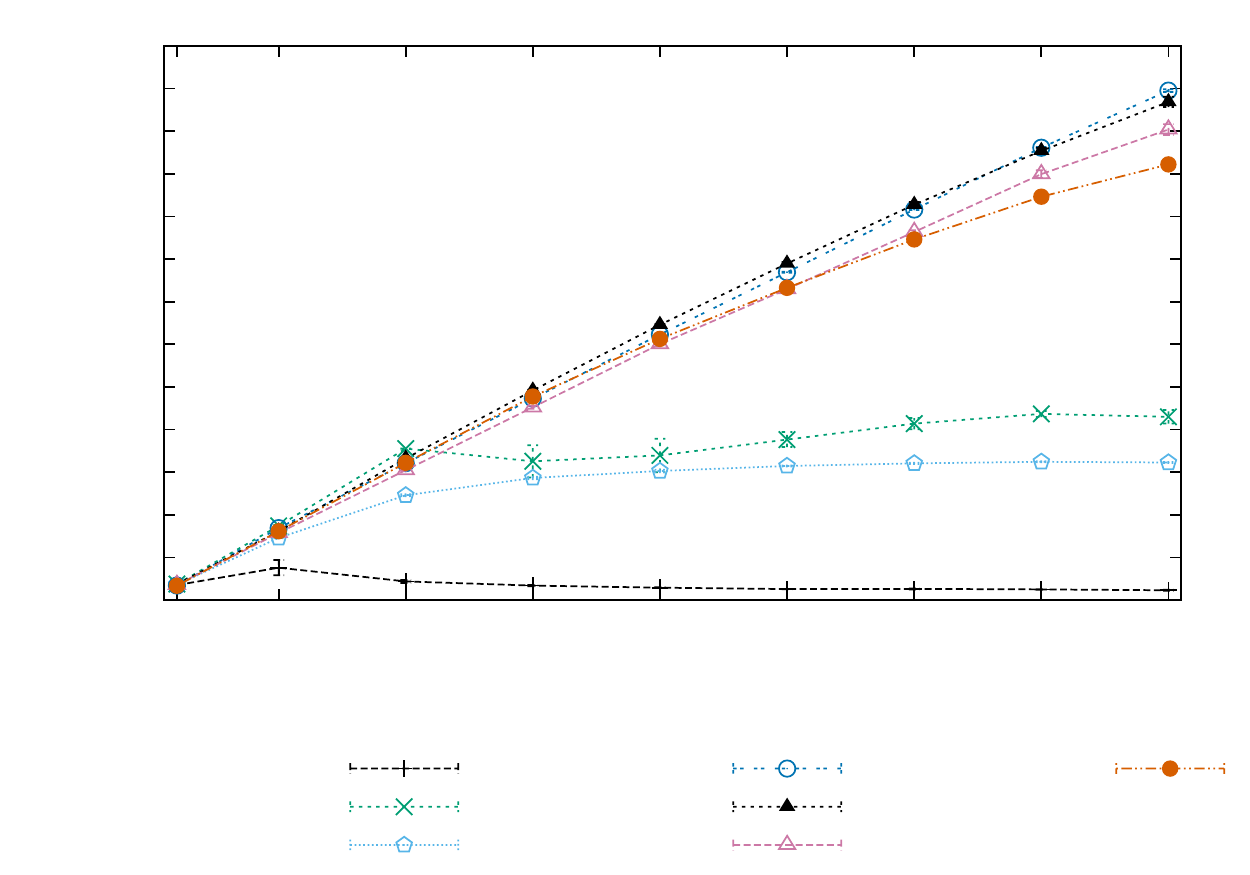}}%
		    \put(6600,3300){\includegraphics[keepaspectratio=true,trim=0 0 200px 210px,clip]{queue-prodcon}}%
		    \put(8395,2625){\includegraphics[keepaspectratio=true,trim=200px 0 100px 210px,clip]{queue-prodcon}}%
		    \put(8190,1950){\includegraphics[keepaspectratio=true,trim=300px 0 0 210px,clip]{queue-prodcon}}%
		    \gplfronttext
		  \end{picture}%
		\endgroup
  }\\
  \centering{``queue-like'' data structures}
	\label{fig:prodcon-queues}
\end{minipage}\\[-3.5em]
\begin{minipage}{\textwidth}
	\centering
	\resizebox{.75\textwidth}{!}{%
		\begingroup
		  \makeatletter
		  \gdef\gplbacktext{}%
		  \gdef\gplfronttext{}%
		  \makeatother
		      \def\colorrgb#1{\color[rgb]{#1}}%
		      \def\colorgray#1{\color[gray]{#1}}%
		      \expandafter\def\csname LTw\endcsname{\color{white}}%
		      \expandafter\def\csname LTb\endcsname{\color{black}}%
		      \expandafter\def\csname LTa\endcsname{\color{black}}%
		      \expandafter\def\csname LT0\endcsname{\color[rgb]{1,0,0}}%
		      \expandafter\def\csname LT1\endcsname{\color[rgb]{0,1,0}}%
		      \expandafter\def\csname LT2\endcsname{\color[rgb]{0,0,1}}%
		      \expandafter\def\csname LT3\endcsname{\color[rgb]{1,0,1}}%
		      \expandafter\def\csname LT4\endcsname{\color[rgb]{0,1,1}}%
		      \expandafter\def\csname LT5\endcsname{\color[rgb]{1,1,0}}%
		      \expandafter\def\csname LT6\endcsname{\color[rgb]{0,0,0}}%
		      \expandafter\def\csname LT7\endcsname{\color[rgb]{1,0.3,0}}%
		      \expandafter\def\csname LT8\endcsname{\color[rgb]{0.5,0.5,0.5}}%
		  \setlength{\unitlength}{0.0500bp}%
		  \begin{picture}(7200.00,5040.00)(1100.00, 1100.00)%
		    \gplbacktext{%
		      \csname LTb\endcsname%
		      \put(814,1804){\makebox(0,0)[r]{\strut{} 0}}%
		      \put(814,2033){\makebox(0,0)[r]{\strut{} 2}}%
		      \put(814,2261){\makebox(0,0)[r]{\strut{} 4}}%
		      \put(814,2490){\makebox(0,0)[r]{\strut{} 6}}%
		      \put(814,2718){\makebox(0,0)[r]{\strut{} 8}}%
		      \put(814,2947){\makebox(0,0)[r]{\strut{} 10}}%
		      \put(814,3175){\makebox(0,0)[r]{\strut{} 12}}%
		      \put(814,3404){\makebox(0,0)[r]{\strut{} 14}}%
		      \put(814,3632){\makebox(0,0)[r]{\strut{} 16}}%
		      \put(814,3861){\makebox(0,0)[r]{\strut{} 18}}%
		      \put(814,4089){\makebox(0,0)[r]{\strut{} 20}}%
		      \put(814,4318){\makebox(0,0)[r]{\strut{} 22}}%
		      \put(814,4546){\makebox(0,0)[r]{\strut{} 24}}%
		      \put(814,4775){\makebox(0,0)[r]{\strut{} 26}}%
		      \put(1019,1584){\makebox(0,0){\strut{}2}}%
		      \put(1605,1584){\makebox(0,0){\strut{}10}}%
		      \put(2337,1584){\makebox(0,0){\strut{}20}}%
		      \put(3069,1584){\makebox(0,0){\strut{}30}}%
		      \put(3801,1584){\makebox(0,0){\strut{}40}}%
		      \put(4533,1584){\makebox(0,0){\strut{}50}}%
		      \put(5266,1584){\makebox(0,0){\strut{}60}}%
		      \put(5998,1584){\makebox(0,0){\strut{}70}}%
		      \put(6730,1584){\makebox(0,0){\strut{}80}}%
		      \put(176,3289){\rotatebox{-270}{\makebox(0,0){\strut{}million operations per sec (more is better)}}}%
		      \put(3874,1254){\makebox(0,0){\strut{}number of threads}}%
		    }%
		    \gplfronttext{%
		      \csname LTb\endcsname%
		      \put(8500,3920){\makebox(0,0)[r]{\strut{}Treiber}}%
		      \csname LTb\endcsname%
		      \put(8500,3700){\makebox(0,0)[r]{\strut{}TS~Stack}}%
		      \csname LTb\endcsname%
		      \put(8500,3480){\makebox(0,0)[r]{\strut{}k-Stack}}%
		      \csname LTb\endcsname%
		      \put(8500,3260){\makebox(0,0)[r]{\strut{}LL+D~Treiber}}%
		      \csname LTb\endcsname%
		      \put(8500,3040){\makebox(0,0)[r]{\strut{}LLD~TS~Stack}}%
		      \csname LTb\endcsname%
		      \put(8500,2820){\makebox(0,0)[r]{\strut{}LLD~k-Stack}}%
		      \csname LTb\endcsname%
		      \put(8500,2600){\makebox(0,0)[r]{\strut{}1-RA~DS}}%
		    }%
		    \gplbacktext
		    \put(0,1000){\includegraphics[keepaspectratio=true,trim=0 50 0 0,clip]{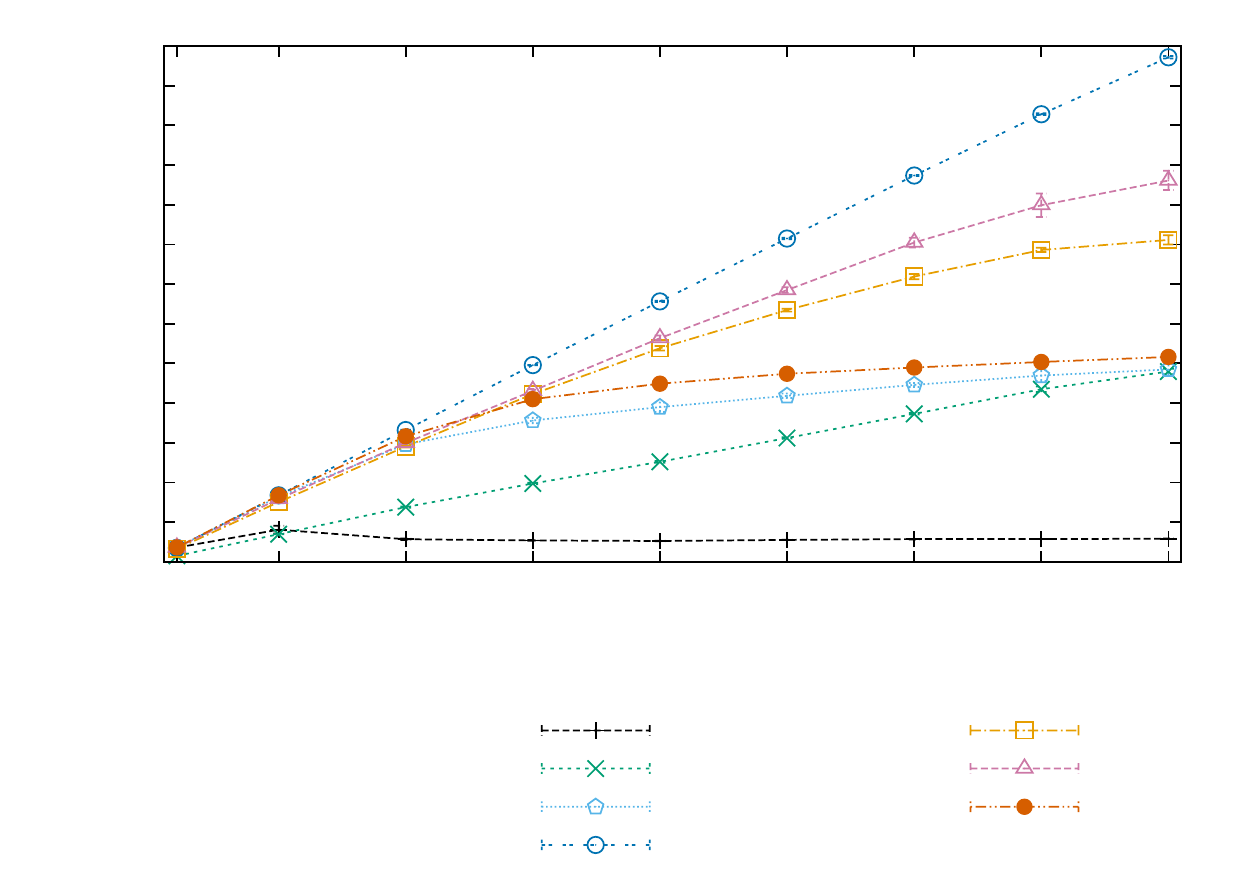}}%
		    \put(5500,3100){\includegraphics[keepaspectratio=true,trim=0 0 100px 200px,clip]{stack-prodcon}}%
		    \put(7025,2200){\includegraphics[keepaspectratio=true,trim=200 0 00px 200px,clip]{stack-prodcon}}%
		    \gplfronttext
		  \end{picture}%
		\endgroup
  }\\
  \centering{``stack-like'' data structures}
	\label{fig:prodcon-stacks}
\end{minipage}
\caption{
Performance and scalability of producer-consumer microbenchmarks with an
increasing number of threads on a 40-core (2 hyperthreads per core) machine}
\label{fig:prodcon}
\end{figure}

Figure~\ref{fig:prodcon} shows the results of the producer-consumer benchmarks.
Similar to experiments performed
elsewhere~\cite{Dodds:POPL15,Henzinger:POPL13,Kirsch:PACT13,Morrison:PPOPP13}
the well-known algorithms MS and Treiber do not scale for 10 or more threads. The state-of-the-art linearizable queue and
stack algorithms LCRQ and TS-interval Stack either perform competitively with
their $k$-out-of-order relaxed counter parts $k$-FIFO and $k$-Stack or even
outperform and outscale them. 
For any implementation~$\Phi$, LLD~$\Phi$ and LL$^+$D~$\Phi$ (when available) perform and scale significantly better than $\Phi$ does, even slightly better than the state-of-the-art pool that we compare to. 
The best improvement show LLD variants of MS~queue and Treiber stack.
The speedup of the locally linearizable implementation to the
fastest linearizable queue (LCRQ) and stack (TS
Stack) implementation at 80 threads is 2.77 and 2.64, respectively.
The performance degradation for LCRQ between 30 and 70 threads aligns with the performance of 
{\tt fetch-and-inc}---the CPU instruction that atomically retrieves and
modifies the contents of a memory location---on the benchmarking
machine, which is different on the original benchmarking machine~\cite{Morrison:PPOPP13}. LCRQ uses {\tt fetch-and-inc} as its key atomic instruction.

\section{Conclusion \& Future Work}
\label{sec:conclusions}

Local linearizability splits a history into a set of thread-induced histories and requires consistency of all such.
This yields an intuitive 
consistency condition for concurrent objects that enables new data structure
implementations with superior performance and scalability.
Local linearizability has desirable properties like compositionality and well-behavedness for container-type data structures.  
As future work, it is interesting to investigate the guarantees that 
local linearizability provides to client programs along the line
of~\cite{Filipovic:TCS10}.

\section*{Acknowledgments}
This work has been supported by the National Research Network RiSE on Rigorous
Systems Engineering (Austrian Science Fund (FWF): S11402-N23, S11403-N23, S11404-N23, S11411-N23), a Google PhD
Fellowship, an Erwin Schr\"odinger Fellowship (Austrian Science Fund (FWF): J3696-N26),
EPSRC grants EP/H005633/1 and EP/K008528/1, the Vienna Science and Technology Fund (WWTF) trough grant PROSEED, the European Research Council (ERC) under grant 267989 (QUAREM) and 
by the Austrian Science Fund (FWF) under grant Z211-N23 (Wittgenstein Award).

\bibliographystyle{plainurl}
\bibliography{p06-haas-arxiv}

\clearpage
\appendix

\section{Local Linearizability with Shape Observers}
\label{sec:LL-shape-obs}

There are two possible ways to deal with shape observers: treat them locally, in the thread-induced history of the performing thread, or treat them globally. While a local treatment is immediate and natural to a local consistency condition, a global treatment requires care. We present both solutions next.

\begin{definition}[Local Linearizability LSO]
\label{def:lllso}
A history $\history{h}$ is locally linearizable with local shape observers (LSO) wrt a sequential specification~$S$ if it is locally linearizable according to Definition~\ref{def:lc} with the difference that the in-methods (Definition~\ref{def:in-out-m}) also contain all shape observers performed by thread $T$, i.e.,
$\mathrm{I}_T = \{m \mid m \in M(\history{h}|{T}) \cap \left(\Ins \cup \SOb\right)\}$.
\qedd
\end{definition}

Global observations require more notation and auxiliary notions.
Let $\history{s_j}$ for $j \in J$ be a collection of sequences over alphabet $\Sigma$ with pairwise disjoint sets of symbols $M(\history{s_j})$. 
A sequence $\history{s}$ is an \emph{interleaving} of $\history{s_j}$ for $j \in J$ if $M(\history{s}) = \bigcup_j M(\history{s_j})$ and $\history{s}| M(\history{s_j}) = \history{s_j}$ for all $j \in J$. 
We write $\prod_{j} \history{s_j}$ for the set of all interleavings of $\history{s_j}$ with $j \in J$. 

Given a history $\history{h}$ and a method call $m \in \history{h}$, we write $\history{h}^{\le m}$ for the (incomplete) history that is the prefix of $\history{h}$ up to and without $m_r$, the response event of $m$. Hence, $\history{h}^{\le m}$ contains all invocation and response events of $\history{h}$ that appear \emph{before} $m_r$. 

\begin{definition}
\label{def:witness}
Let $S$ denote the sequential specification of a container $D$.
A shape observer $m$ in a history $\history{h}$ has a witness if there exists a sequence $\history{s} \in \Sigma^*$ such that $\history{s}m \in S$ and $\history{s} \in \prod_{T} \history{s_T}$ for some $\history{s}_T$ that is a linearization of the thread-induced history $(\history{h}^{\le m})_T$.\qedd
\end{definition}

Informally, the above definition states that a global shape observer $m$ must be justified by a (global) witness. 
Such a global witness is a sequence that (1) when extended by $m$ belongs to the sequential specification, and (2) is an interleaving of linearizations of the thread-induced histories up to $m$.

\begin{definition}[Local Linearizability GSO]
\label{def:llgso}
A history $\history{h}$ is locally linearizable with global shape observers (GSO) wrt a sequential specification~$S$ if it is locally linearizable and
each shape observer $m \in \SOb$ has a witness.
\qedd
\end{definition}

We illustrate the difference in the local vs.~the global approach for shape observers with the following example.
\begin{example}\label{example:LSO-GSO}
Consider the following queue history with global observer $\size{}$
	$$\begin{minipage}{0.5\textwidth}
	\centering
	\scalebox{1}{
	\begin{tikzpicture}
							
		\node (T1) at (0.25,.75) {$T_1$};
		\node (T2) at (0.25,0) {$T_2$};
							
		\draw[dotted,->] (0.75, .75) -- (7.75, .75);
		\draw[dotted,->] (0.75, 0) -- (7.75, 0);
		\draw[|-|] (1,.75) -- node[above] {$\enq{1}$} (2,.75);
		\draw[|-|] (1.5,0) -- node[above] {$\deq{1}$} (4,0);
		\draw[|-|] (3,.75) -- node[above] {$\enq{2}$} (6,.75);
		\draw[|-|] (5,0) -- node[above] {$\size{n}$} (7,0);
							
	\end{tikzpicture}
	}
	\end{minipage}$$
where $n$ is just a placeholder for a concrete natural number.
For $n = 0$, the history \history{h} is locally linearizable LSO, but not locally linearizable GSO.
For $n = 1$, the history \history{h} is locally linearizable GSO, but not locally linearizable LSO. 	
\end{example}

Global observers and non-disjoint operations are expected to have negative impact on performance. 
If one cares for global consistency, local linearizability is not the consistency condition to be used. 
The restriction to containers and disjoint operations specifies, in an informal way, the minimal requirements for local consistency to be acceptable.

Neither sets nor maps are containers according to our definition. 
However, it is possible to extend our treatment to sets and maps similar to our treatment of global observers. 
Locally linearizable sets and maps will be weaker than their linearizable counterparts, but, due to the tight coupling between mutator and observer effects, the gain in performance is unlikely to be as substantial as the one observed in other data structures. 
The technicalities needed to extend local linearizability to sets and maps would complicate the theoretical development without considerable benefits and we, therefore, excluded such data structures.

\section{Additional Results and Proofs}
\label{sec:proofs}

\OurTheorem{\ref{thm:comp}}{Compositionality}{
A history $\history{h}$ over a set of objects~$O$ with sequential 
specifications~$S_q$ for $q \in O$ is locally linearizable if and only if 
$\history{h}|q$ is locally linearizable with respect to~$S_q$ for 
every~$q \in O$.
}

\begin{proof}
The property follows from the compositionality of linearizability and the fact 
that $(\history{h}|q)_T = \history{h}_T|q$ for every thread~$T$ and 
object~$q$. 
Assume that $\history{h}$ over $O$ is locally linearizable.
This means that all thread-induced histories~$\history{h}_T$ over $O$ are 
linearizable.
Hence, since linearizability is compositional, for each object~$q \in O$ the 
history~$\history{h}_T|q$ is linearizable with respect to~$S_q$.
Now from $(\history{h}|q)_T = \history{h}_T|q$ we have that for every 
object~$q$ the history~$(\history{h}|q)_T$ is linearizable for every 
thread~$T$.

Similarly, assume that for every object~$q \in O$ the history~$\history{h}|q$ 
is locally linearizable. 
Then, for every~$q$, $(\history{h}|q)_T = \history{h}_T|q$ is linearizable for 
every thread~$T$. 
From the compositionality of linearizability, $\history{h}_T$ is linearizable 
for every thread~$T$. 
This proves that $\history{h}$ is locally linearizable.
\end{proof}

\OurProposition{\ref{prop:lin-loc2}}{Lin vs. LocLin 2}{
Linearizability implies local linearizability for sequential specifications that
are closed under data-projection.
}

\begin{proof}
Assume we are given a history~$\history{h}$ which is linearizable 
with respect to a sequential specification~$S$ that is closed under data-projection. 
Further assume that, without loss of generality, $\history{h}$ is complete.
Then there exists a sequential history~$\history{s} \in S$ such 
that
\begin{inparaenum}[(1)]
	\item $\history{s}$ is a permutation of $\history{h}$, and
	\item if $m \precOrder{\history{h}} n$, then also $m \precOrder{\history{s}} n$.
\end{inparaenum}
Given a thread~$T$, consider the thread-induced history $\history{h}_T$
and let $\history{s}_T = \history{s}| 
\left(\mathrm{I}_T \cup \mathrm{O}_T\right)$.
Then, $\history{s}_T$ is a permutation of $\history{h}_T$ since $\history{h}_T$ 
and $\history{s}_T$ consist of the same events.
Furthermore, $\history{s}_T \in S$ 
since $S$ is closed under data-projection and since Equation~(\ref{eq:cont-obs}) holds for containers.
Finally, we have for each $m\in\history{h}_T$ and $n\in\history{h}_T$ that,
if $m \precOrder{\history{h}_T} n$, then also $m \precOrder{\history{s}_T} n$ 
since $m \precOrder{\history{h}} n$ and therefore $m \precOrder{\history{s}} n$ which implies $m \precOrder{\history{s}_T} n$.
Thereby, we have shown that $\history{h}_T$ is linearizable with respect 
to~$S$, for an arbitrary thread~$T$. 
Hence $\history{h}$ is locally linearizable with respect to~$S$.
\end{proof}

\OurProposition{\ref{prop:DS-are-closed-under-projection}}{Data-Projection Closedness}{
The sequential specifications of pool, queue, and stack are closed under data-projection.
}

\begin{proof}
Let $\history{s} \in S_P$, $V' \subseteq V$, and let $$\history{s'} = \history{s}|\left(\{ \ins{x}, \rem{x} \mid x \in V' \cup\Emp\}\right).$$ 
Then, it suffices to check that all axioms for pool (Definition~\ref{example:pool} and Table~\ref{tab:axioms}) hold for $\history{s'}$.
Clearly, all methods in $\history{s'}$ appear at most once, as they do so in $\history{s}$.
If $\rem{x} \in \history{s'}$, then $\rem{x} \in \history{s}$ and, since $\history{s} \in S_P$, $\ins{x} \abOrder{\history{s}} \rem{x}$. But then also $\rem{x} \in \history{s'}$ and hence $\ins{x} \abOrder{\history{s'}} \rem{x}$. Finally, if $\ins{x} \abOrder{\history{s'}} \rem{e}$ for $e \in \Emp$, then $\ins{x} \abOrder{\history{s}} \rem{e}$ implying that $\rem{x} \in \history{s}$ and $\rem{x}\abOrder{\history{s}} \rem{e}$. But then $\rem{x} \in \history{s'}$ as well and $\rem{x}\abOrder{\history{s'}} \rem{e}$. This shows that $S_P$ is closed under data-projection.

Assume now that $\history{s} \in S_Q$ and $\history{s'}$ is as before (with $\enq{}$ and $\deq{}$ for $\ins{}$ and $\rem{}$, respectively). Then, as $S_P$ is closed under data-projection, $\history{s'}$ satisfies the pool axioms. Moreover, the queue-order axiom (Definition~\ref{example:queue} and Table~\ref{tab:axioms}) also holds: Assume $\enq{x} \abOrder{\history{s'}} \enq{y}$ and $\deq{y} \in \history{s'}$. Then $\enq{x} \abOrder{\history{s}} \enq{y}$ and $\deq{y} \in \history{s}$. Since $\history{s} \in S_Q$ we get $\deq{x} \in \history{s}$ and $\deq{x} \abOrder{\history{s}} \deq{y}$. But this means $\deq{x} \in \history{s'}$ and $\deq{x} \abOrder{\history{s'}} \deq{y}$. Hence, $S_Q$ is closed under data-projection.

Finally, if $\history{s} \in S_S$ and $\history{s'}$ is as before (with $\push{}$ and $\pop{}$ for $\ins{}$ and $\rem{}$, respectively), we need to check that the stack-order axiom (Definition~\ref{example:stack} and Table~\ref{tab:axioms}) holds. Assume $\push{x} \abOrder{\history{s'}} \push{y} \abOrder{\history{s'}} \pop{x}$. This implies $\push{x} \abOrder{\history{s}} \push{y} \abOrder{\history{s}} \pop{x}$ and since $\history{s} \in S_S$ we get $\pop{y} \in \history{s}$ and $\pop{y} \abOrder{\history{s}} \pop{x}$. But then $\pop{y} \in \history{s'}$ and $\pop{y} \abOrder{\history{s'}} \pop{x}$. So, $S_S$ is closed under data-projection.
\end{proof}

\OurProposition{\ref{prop:LL-pool}}{LocLin Pool}{
Let $\history{h}$ be a locally linearizable history wrt a pool. Then:
\begin{enumerate}[1.]
\item No value is duplicated, i.e., every remove method appears in $\history{h}$ at most once.
\item There are no out-of-thin-air values, i.e., $$\forall x \in V. \,\, \rem{x} \in \history{h} \Rightarrow \ins{x} \in \history{h} \,\,\,\wedge\,\,\, \rem{x} {\not<}_{\history{h}} \ins{x}.$$
\item No value is lost, i.e., 
$\forall x \in V. \,\, \forall e \in \Emp. \,\, \ins{x} \precOrder{\history{h}} \rem{e}
\Rightarrow \rem{x} \in \history{h} \,\,\,\wedge\,\,\, \rem{e} {\not<}_{\history{h}} \rem{x}$ and 
$\forall x \in V. \,\,\forall e \in \Emp.\,\,\rem{e} \precOrder{\history{h}} \rem{x} \Rightarrow \ins{x} \not<_{\history{h}} \rem{e}.$
\end{enumerate}
}

\begin{proof}
Note that if a history $\history{h}$ is linearizable wrt a pool, then all of the three stated properties hold, as a consequence of linearizability and the definition of $S_P$. Now assume that $\history{h}$ is locally linearizable wrt a pool.

If $\rem{x}$ appears twice in $\history{h}$, then it also appears twice in some thread-induced history $\history{h}_T$ contradicting that $\history{h}_T$ is linearizable with respect to a pool. This shows that no value is duplicated.

If $\rem{x} \in \history{h}$, then $\rem{x}\in\history{h}_T$ for some $T$ and, since $\history{h}_T$ is linearizable with respect to a pool, $\ins{x} \in \history{h}_T$ and $\rem{x} {\not<}_{\history{h}_T} \ins{x}$. This yields $\ins{x} \in \history{h}$ and $\rem{x} {\not<}_{\history{h}} \ins{x}$. Hence, there are no thin-air values.

Finally, if $\rem{e}\in\history{h}$ for $e \in \Emp$ then $\rem{e}\in\history{h}_T$ for all $T$. Let $\ins{x} \precOrder{\history{h}} \rem{e}$ and let $T'$ be such that $\ins{x}\in\history{h}_{T'}$. Then $\ins{x} \precOrder{\history{h}_{T'}} \rem{e}$ and since $\history{h}_{T'}$ is linearizable with respect to a pool, $\rem{x} \in \history{h}_{T'}$ and $\rem{e} {\not<}_{\history{h}_{T'}} \rem{x}$. This yields $\rem{x} \in \history{h}$ and $\rem{e} {\not<}_{\history{h}} \rem{x}$. Similarly, the other condition holds. Hence, no value is lost.
\end{proof}

\OurTheorem{\ref{thm:queue}}{Queue Local Linearizability}{
A queue concurrent history $\history{h}$ is locally linearizable with respect to the queue sequential specification $S_Q$ if and only if 
\begin{enumerate}[1.]

	\item $\history{h}$ is locally linearizable with respect to the pool sequential specification $S_P$, and
	
	\item $\forall x,y \in V. \,\,\,\forall i. \,\,\,\enq{x} <_{\history{h}}^i \enq{y} \,\,\,\wedge\,\,\, \deq{y} \in \history{h} \,\,\,\Rightarrow\,\,\, \deq{x} \in \history{h} \,\,\,\wedge\,\,\,\deq{y} \not<_\history{h} \deq{x}$.
\end{enumerate}
}

\begin{proof} 
Assume $\history{h}$ is locally linearizable with respect to $S_Q$. 
Since $S_Q \subseteq S_P$ (with suitably renamed method calls), $\history{h}$ 
is locally linearizable with respect to $S_P$. 
Moreover, since all $\history{h}_i$ are linearizable with respect to $S_Q$, by 
Theorem~\ref{thm:queue-lin}, for all $i$ we have  
$\forall x,y \in V. \,\,\, \enq{x} <_{\history{h}_i} \enq{y} \,\,\,\wedge\,\,\, \deq{y} \in \history{h}_i \,\,\,\Rightarrow\,\,\, \deq{x} \in \history{h}_i \,\,\,\wedge\,\,\,\deq{y} \not<_{\history{h}_i} \deq{x}.$

Assume $x,y \in V$ are such that $\enq{x} <_{\history{h}}^i \enq{y}$ and 
$\deq{y} \in \history{h}$. 
Then $\enq{x} <_{\history{h}_i} \enq{y}$ and $\deq{y} \in \history{h}_i$ so 
$\deq{x} \in \history{h}_i$ and $\deq{y} \not<_{\history{h}_i} \deq{x}$. 
This implies $\deq{x} \in \history{h}$ and $\deq{y} \not<_\history{h} \deq{x}$.

For the opposite, assume that conditions 1. and 2. hold for a 
history~$\history{h}$. 
We need to show that (1) $\history{h}_i$ form a decomposition of 
$\history{h}$, which is clear for a queue, and (2) each $\history{h}_i$ is 
linearizable with respect to $S_Q$. 

By 1., each $\history{h}_i$ is linearizable with respect to a pool. 
Assume $\enq{x} <_{\history{h}_i} \enq{y}$ and $\deq{y} \in \history{h}_i$. 
Then $\enq{x} <_{\history{h}}^i \enq{y} \,\,\,\wedge\,\,\, \deq{y} \in \history{h}$ 
and hence by 2., $\deq{x} \in \history{h} \,\,\,\wedge\,\,\,\deq{y} \not<_\history{h} \deq{x}$. 
Again, as $\enq{x}, \deq{x} \in \history{h}_i$ we get $\deq{x} \in \history{h}_i \,\,\,\wedge\,\,\,\deq{y} \not<_{\history{h}_i} \deq{x}$. 
According to Theorem~\ref{thm:queue-lin} this is enough to conclude that 
each $\history{h}_i$ is linearizable with respect to $S_Q$.
\end{proof}

\OurTheorem{\ref{thm:sc:queues-and-stacks}}{Pool, Queue, \& Stack, SC}{For pool, queue, and stack, local linearizability is incomparable to sequential consistency.}

\begin{proof} 
The following histories, when instantiating $\ii{}$ with $\ins{}$, $\enq{}$, and $\push{}$, respectively, and instantiating $\rr{}$ with $\rem{}$, $\deq{}$, and $\pop{}$, respectively, are sequentially consistent but not locally linearizable wrt pool, queue and stack:
\begin{itemize}
\item[(a)] Pool: 
\begin{center}
		\begin{tikzpicture}
		
			\node (T1) at (0.25,.75) {$T_1$};
			\node (T2) at (0.25,0) {$T_2$};
		
			\draw[dotted,->] (0.75, .75) -- (5, .75);
			\draw[dotted,->] (0.75, 0) -- (5, 0);
			\draw[|-|] (1,.75) -- node[above] {$\ii{1}$} (1.75,.75);
			\draw[|-|] (4,.75) -- node[above] {$\rr{1}$} (4.75,.75);
			\draw[|-|] (2,0) -- node[above] {$\rr{\emptyValue}$} (3.75,0);
			
		\end{tikzpicture}
\end{center}
\item[(b)] Queue:
	\begin{center}
		\begin{tikzpicture}
										
			\node (T1) at (0.25,.75) {$T_1$};
			\node (T2) at (0.25,0) {$T_2$};
									
			\draw[dotted,->] (0.75, .75) -- (5, .75);
			\draw[dotted,->] (0.75, 0) -- (5, 0);
			\draw[|-|] (1,0) -- node[above] {$\ii{1}$} (1.75,0);
			\draw[|-|] (3,.75) -- node[above] {$\rr{2}$} (3.75,.75);
			\draw[|-|] (2,0) -- node[above] {$\ii{2}$} (2.75,0);
			\draw[|-|] (4,0) -- node[above] {$\rr{1}$} (4.75,0);
									
		\end{tikzpicture}
	\end{center}
\item[(c)] Stack:
	\begin{center}
		\begin{tikzpicture}
										
			\node (T1) at (0.25,.75) {$T_1$};
			\node (T2) at (0.25,0) {$T_2$};
									
			\draw[dotted,->] (0.75, .75) -- (5, .75);
			\draw[dotted,->] (0.75, 0) -- (5, 0);
			\draw[|-|] (1,.75) -- node[above] {$\ii{1}$} (1.75,.75);
			\draw[|-|] (3,0) -- node[above] {$\rr{1}$} (3.75,0);
			\draw[|-|] (2,.75) -- node[above] {$\ii{2}$} (2.75,.75);
			\draw[|-|] (4,0) -- node[above] {$\rr{2}$} (4.75,0);
										
		\end{tikzpicture}
	\end{center}
\end{itemize}

History (a) is already not locally linearizable wrt pool, queue, and stack, respectively, histories (b) and (c) provide interesting examples.
The history in Figure~\ref{fig:ll-notsc} is locally 
linearizable but not sequentially consistent wrt a pool.
The following histories are locally linearizable but not sequentially 
consistent wrt a queue and a stack, respectively:
\begin{enumerate}
\item[(d)] Queue:
	\begin{center}
	\scalebox{1}{
		\begin{tikzpicture}
		
			\node (T1) at (0.25,.75) {$T_1$};
			\node (T2) at (0.25,0) {$T_2$};
		
			\draw[dotted,->] (0.75, .75) -- (9, .75);
			\draw[dotted,->] (0.75, 0) -- (9, 0);
			\draw[|-|] (1,.75) -- node[above] {$\ii{1}$} (1.75,.75);
			\draw[|-|] (2,.75) -- node[above] {$\ii{2}$} (2.75,.75);
			\draw[|-|] (3,.75) -- node[above] {$\ii{3}$} (3.75,.75);
			\draw[|-|] (4,.75) -- node[above] {$\rr{1}$} (4.75,.75);
			\draw[|-|] (5,0) -- node[above] {$\rr{2}$} (5.75,0);
			\draw[|-|] (6,0) -- node[above] {$\ii{4}$} (6.75,0);
			\draw[|-|] (7,0) -- node[above] {$\rr{4}$} (7.75,0);
			\draw[|-|] (8,0) -- node[above] {$\rr{3}$} (8.75,0);
			
		\end{tikzpicture}
		}
	\end{center}
\end{enumerate}
The two thread-induced histories $\ii{1}\ii{2}\ii{3}\rr{1}\rr{2}\rr{3}$ 
and $\ii{4}\rr{4}$ are both linearizable with respect to a queue.
However, the overall history has no sequential witness and is therefore not 
sequentially consistent:
To maintain the queue behavior, the order of operations $\rr{1}$ and $\rr{2}$ 
cannot be changed.
However, this implies that the value~$3$ instead of the value~$4$ would have 
to be removed directly after $\ii{4}$.
\begin{enumerate}
\item[(e)] Stack:
	\begin{center}
		\begin{tikzpicture}
										
			\node (T1) at (0.25,.75) {$T_1$};
			\node (T2) at (0.25,0) {$T_2$};								
											
			\draw[dotted,->] (0.75, .75) -- (7, .75);
			\draw[dotted,->] (0.75, 0) -- (7, 0);
			\draw[|-|] (1,.75) -- node[above] {$\ii{1}$} (1.75,.75);
			\draw[|-|] (2,.75) -- node[above] {$\ii{2}$} (2.75,.75);
			\draw[|-|] (3,0) -- node[above] {$\rr{2}$} (3.75,0);
			\draw[|-|] (4,0) -- node[above] {$\ii{3}$} (4.75,0);
			\draw[|-|] (5,0) -- node[above] {$\rr{1}$} (5.75,0);
			\draw[|-|] (6,0) -- node[above] {$\rr{3}$} (6.75,0);
										
		\end{tikzpicture}
	\end{center}
\end{enumerate}
The two thread-induced histories $\ii{1}\ii{2}\rr{2}\rr{1}$ and 
$\ii{3}\rr{3}$ are both linearizable with respect to a stack.
The operations $\ii{2}$ and $\rr{2}$ prevent the reordering of operations
$\ii{1}$ and $\ii{3}$.
Therefore, the overall history has no sequential witness and hence it is not 
sequentially consistent.
\end{proof}

\OurProposition{\ref{prop:pool-qc}}{Pool, QC}{Let $\history{h}$ be a pool history in which no data is duplicated, no thin-air values are returned, and no data is lost, i.e., $\history{h}$ satisfies 1.-3. of Proposition~\ref{prop:LL-pool}. Then $\history{h}$ is quiescently consistent.}

\begin{proof}
Assume $\history{h}$ is a pool history that satisfies 1.-3. of Proposition~\ref{prop:LL-pool}. 
Let $\history{h}_1, \dots, \history{h}_n$ be histories that form a \emph{sequential decomposition} of $\history{h}$. 
That is $\history{h} = \history{h}_1 \cdots \history{h}_n$ and the only quiescent states in any $\history{h}_i$ are at the beginning and at the end of it. 
Note that this decomposition has nothing to do with a thread-local decomposition. 
Let $M_i = M_{\history{h}_i}$ be the set of methods of $\history{h}_i$, for $i \in \{1,\dots, n\}$. 
Note that the sanity conditions 1.-3. ensure that none of the following two situations can happen:
\begin{itemize}
\item $\rem{x} \in M_i, \ins{x} \in M_j, j > i$, 
\item $\ins{x} \in M_i, \rem{\emptyValue} \in M_j, \rem{x} \in M_k, k > j > i$,
\end{itemize}
Let $V_i = \{x_{i,1}, \dots, x_{i,m}\}$ denote the set of values in $M_i$ ordered in a way that there is a $p$ and $q$ such that
\begin{itemize}
\item $\ins{x_{i,j}}, \rem{x_{i,j}} \in M_i$ for $j \le p$;
\item $\rem{x_{i,j}} \in M_i$ for $j > p, j \le q$; and
\item $\ins{x_{i,j}} \in M_i$ for $j > q$.
\end{itemize}
Moreover, let $r_i$ be the number of occurrences of $\rem{\emptyValue}$ in $\history{h}_i$.

We now construct a sequential history for $\history{h}$, which has the form $ \history{q} =\history{q}_1\cdots \history{q}_n$ where each sequential history $\history{q}_i$ is a permutation of $M_i$ shown in Figure~\ref{appendix:fix:A0}. 
Using the observations above, it is easy to check that $\history{q}$ is indeed a quiescent witness for $\history{h}$.
\end{proof}

\begin{figure}[tb]
$\history{q}_i = \ins{x_{i,1}}\rem{x_{i,1}}\dots \ins{x_{i,p}}\rem{x_{i,p}}\rem{x_{i,p+1}}\rem{x_{i,q}}\rem{\emptyValue}^{r_i} \ins{x_{i,q+1}} \dots \ins{x_{i,m}}.$
\caption{Sequential history $\history{q}_i$.}
\label{appendix:fix:A0}
\end{figure}

\OurTheorem{\ref{thm:qc:pool-queue-stack}}{Pool, Queue, \& Stack, QC}{For pool, local linearizability is stronger than quiescent consistency. For queue and stack, local linearizability is incomparable to quiescent consistency.}

\begin{proof} 
The following histories are quiescently consistent but not locally linearizable wrt
pool, queue, and stack, respectively:\\[1em]
\begin{minipage}{\columnwidth}
\begin{itemize}
\item[(a)] Pool:
	\begin{center}
	\scalebox{1}{
		\begin{tikzpicture}
		
			\node (T1) at (0.25,.75) {$T_1$};
			\node (T2) at (0.25,0) {$T_2$};
		
			\draw[dotted,->] (0.75, .75) -- (7.5, .75);
			\draw[dotted,->] (0.75, 0) -- (7.5, 0);
			\draw[|-|] (1,.75) -- node[above] {$\rem{\emptyValue}$} (6.75,.75);
			\draw[|-|] (1.25,0) -- node[above] {$\ins{1}$} (2.5,0);
			\draw[|-|] (2.75,0) -- node[above] {$\rem{\emptyValue}$} (5,0);
			\draw[|-|] (5.25,0) -- node[above] {$\rem{1}$} (6.5,0);
			
		\end{tikzpicture}
	}
	\end{center}
\end{itemize}
\end{minipage}
\\[1em]
\begin{minipage}{\columnwidth}
\begin{itemize}
\item[(b)] Queue:
	\begin{center}
	\scalebox{1}{
		\begin{tikzpicture}
			
			\node (T1) at (0.25,.75) {$T_1$};
			\node (T2) at (0.25,0) {$T_2$};
			
			\draw[dotted,->] (0.75, .75) -- (7.5, .75);
			\draw[dotted,->] (0.75, 0) -- (7.5, 0);
			\draw[|-|] (1,.75) -- node[above] {$\enq{1}$} (7.25,.75);
			\draw[|-|] (1.25,0) -- node[above] {$\enq{2}$} (2.5,0);
			\draw[|-|] (2.75,0) -- node[above] {$\enq{3}$} (4,0);
			\draw[|-|] (4.25,0) -- node[above] {$\deq{3}$} (5.5,0);
			\draw[|-|] (5.75,0) -- node[above] {$\deq{2}$} (7,0);
				
		\end{tikzpicture}
	}
	\end{center}
\end{itemize}
\end{minipage}
\\[1em]
\begin{minipage}{\columnwidth}
\begin{itemize}
\item[(c)] Stack:
	\begin{center}
	\scalebox{1}{
		\begin{tikzpicture}
										
			\node (T1) at (0.25,.75) {$T_1$};
			\node (T2) at (0.25,0) {$T_2$};
									
			\draw[dotted,->] (0.75, .75) -- (7.5, .75);
			\draw[dotted,->] (0.75, 0) -- (7.5, 0);
			\draw[|-|] (1,.75) -- node[above] {$\push{1}$} (7.25,.75);
			\draw[|-|] (1.25,0) -- node[above] {$\push{2}$} (2.5,0);
			\draw[|-|] (2.75,0) -- node[above] {$\push{3}$} (4,0);
			\draw[|-|] (4.25,0) -- node[above] {$\pop{2}$} (5.5,0);
			\draw[|-|] (5.75,0) -- node[above] {$\pop{3}$} (7,0);
													
		\end{tikzpicture}
	}
	\end{center}
\end{itemize}
\end{minipage}\\[1em]

In all three histories, the only quiescent states are before and after the 
longest operation.
Therefore, all operations in thread~$T_2$ can be reordered arbitrarily, in 
particular in a way such that they satisfy the sequential specification of
the respective concurrent data structure.
However, each of the thread-induced histories for thread~$T_2$ are not 
linearizable with respect to pool, queue, and stack, respectively.
Therefore, none of these histories is locally linearizable. Also here history (a) suffices.

On the other hand, the following histories are not quiescently consistent but
locally linearizable wrt queue, and stack, respectively:
\begin{itemize}
\item[(d)] Queue:
	\begin{center}
	\scalebox{1}{
		\begin{tikzpicture}
			
			\node (T1) at (0.25,.75) {$T_1$};
			\node (T2) at (0.25,0) {$T_2$};
			
			\draw[dotted,->] (0.75, .75) -- (7, .75);
			\draw[dotted,->] (0.75, 0) -- (7, 0);
			\draw[|-|] (1,.75) -- node[above] {$\enq{1}$} (2.25,.75);
			\draw[|-|] (2.5,0) -- node[above] {$\enq{2}$} (3.75,0);
			\draw[|-|] (4,.75) -- node[above] {$\deq{2}$} (5.25,.75);
			\draw[|-|] (5.5,0) -- node[above] {$\deq{1}$} (6.75,0);
				
		\end{tikzpicture}
	}
	\end{center}
\item[(e)] Stack:
	\begin{center}
	\scalebox{1}{
		\begin{tikzpicture}
										
			\node (T1) at (0.25,.75) {$T_1$};
			\node (T2) at (0.25,0) {$T_2$};
									
			\draw[dotted,->] (0.75, .75) -- (7, .75);
			\draw[dotted,->] (0.75, 0) -- (7, 0);
			\draw[|-|] (1,.75) -- node[above] {$\push{1}$} (2.25,.75);
			\draw[|-|] (2.5,0) -- node[above] {$\push{2}$} (3.75,0);
			\draw[|-|] (4,.75) -- node[above] {$\pop{1}$} (5.25,.75);
			\draw[|-|] (5.5,0) -- node[above] {$\pop{2}$} (6.75,0);													
		\end{tikzpicture}
	}
	\end{center}
\end{itemize}

In histories~(d) and~(e), between each two operations, the concurrent data 
structure is in a quiescent state.
Therefore, none of the operations can be reordered and, hence, no sequential
witness exists.
However, all thread-induced histories are linearable and, 
therefore, the overall histories are locally linearizable.
In particular, on a history where each pair of operations is separated by a 
quiescent state, i.e., there is no overlap of operations, a quiescent 
consistent data structure behaves as it would be linearizable with respect to its 
sequential specification and we see the same semantic differences to local
linearizability as we see between linearizability and local linearizability.
\end{proof}

\section{Case Study: Work Stealing Queues} 
\label{sec:WSQ}

Consider a data structure $D$ which admits two operation types: $\ins{x}$, which
inserts the element $x$ into the container, and $\rem{}$, which returns and
removes an element from the container.
Now imagine that the implementation uses a Work Stealing Queue (WSQ)~\cite{Michael:PPOPP09}.
Every thread $T$ that uses $D$ has its unique designated buffer $Q_T$ in the WSQ.
Whenever thread $T$ calls $\ins{x}$, $x$ is appended to the tail of $Q_T$.
When $T$ calls $\rem{}$, WSQ first checks whether $Q_T$ is non-empty; if it is,
then it returns the element at the tail of $Q_T$ (LIFO semantics) and removes
it.
Otherwise, it chooses some other $Q_{T'}$ and tries to return an element from
that buffer.
But any time a different thread's buffer is checked, the element to be removed is
taken from the head (FIFO semantics).
If $T$ and $T'$ are both trying to access the same buffer at the same time, then
usual synchronization measures are taken to ensure that exactly one thread
removes one element.

Given this implementation, the developer of $D$ wants to write a specification
for the potential users of $D$.
Since $D$ is essentially a collection of deques, the developer is tempted to
state that $D$ is a deque with a particular consistency condition.
However, $D$ is not a linearizable deque because $\ins{x}$ by $T$ followed by
$\ins{y}$ by $T'$ followed by $\rem{}$ returns either $x$ or $y$ depending on
whether $T$ or $T'$ calls it; i.e. $\rem{}$ has ambiguous semantics.
$D$ can be seen as a sequentially consistent (SC) deque but then $D$ does not allow
many behaviors that an SC deque would allow; i.e. SC does not capture the
behaviors of $D$ tightly.
Relaxed sequential specifications will not work either since $D$ does converge
to sequential semantics (of a LIFO stack) when a single thread uses it.
In short, the developer will fail to capture the semantics of $D$ in a
satisfactory manner.

$D$ on the other hand is a locally linearizable deque in which $\rem{}$ by $T$
from $Q_{T'}$ is treated as FIFO removal whenever $T\neq T'$ and as LIFO
removal whenever $T=T'$.
In other words, local linearizability provides a succinct and clean representation of a
well-known implementation framework (WSQ) hiding away implementation details.
Compare this with the fact that even though WSQ has a {\em queue} in it, to
argue its {\em correctness} it is proved to be a linearizable pool even though
it has stronger semantics than a pool; i.e. linearizable pool semantics is too
weak for $D$.
Observe also that since what we have described in the example is essentially providing the
illusion of using a monolithic structure which is implemented in terms of
distributed components (shared memory is typically implemented on message
passing), we expect local linearizability to be widely applicable.

\section{Quiescent Consistency \& Quantitative Quiescent Consistency}
\label{appendix:qqc}

Without going into the details of the definition of quantitative quiescent 
consistency we give a history in Figure~\ref{fig:qqc-notll} that is quantitatively 
quiescently consistent but not locally linearizable wrt a queue.
Quantitative quiescent consistency allows to reorder the two insert-operations
in thread~$T_2$ and thereby violates local linearizability.

\begin{figure}[t]
	\centering
	\scalebox{1}{
	  \begin{tikzpicture}
		
		  \node (T1) at (0.25,.75) {$T_1$};
		  \node (T2) at (0.25,0) {$T_2$};
	
		  \draw[dotted,->] (0.75, .75) -- (7.75, .75);
		  \draw[dotted,->] (0.75, 0) -- (7.75, 0);
	  	\draw[|-|] (1,.75) -- node[above] {$\enq{3}$} (3.75,.75);
  		\draw[|-|] (1.25,0) -- node[above] {$\enq{2}$} (2.25,0);
		  \draw[|-|] (2.5,0) -- node[above] {$\enq{1}$} (3.5,0);
		  \draw[|-|] (4,0) -- node[above] {$\deq{1}$} (5,0);
	  	\draw[|-|] (5.25,0) -- node[above] {$\deq{2}$} (6.25,0);
  		\draw[|-|] (6.5,0) -- node[above] {$\deq{3}$} (7.5,0);
		
	  \end{tikzpicture}
	}
	\caption{History that is QQC but not LL.}
	\label{fig:qqc-notll}
\end{figure}
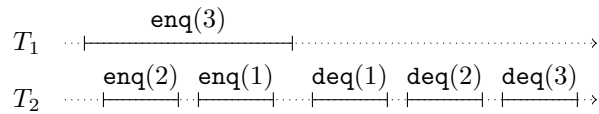

\section{Consistency Conditions for Distributed Shared Memory}
\label{appendix:ccdsm}

\begin{landscape}
\begin{table}
\begin{tabular}{l@{\hspace{1.7em}}l@{\hspace{1.7em}}c@{\hspace{1.7em}}l@{\hspace{1.7em}}l@{\hspace{1.7em}}l@{\hspace{1.7em}}c}
\toprule

\textbf{Consistency} &                                  
\textbf{Decomposition per} &                                
\textbf{\#SHs} &                                         
\textbf{Write-Operations} \textbf{I}$_\history{h}(i)$ & 
\textbf{Read-Operations} \textbf{O}$_\history{h}(i)$ &  
\textbf{CCfSH} &                                         
\textbf{LoD}                                            
\\

\textbf{Condition} &                                    
&                                          
\\

\midrule
LL        & thread              & $n$         & $\{ \op{ins}(v) \in \history{h}|T_i \mid v \in V \}$ & $\{\op{head}(v_{init}) \in \history{h}\}$ & Lin.   & no \\
          &                     &             & & $\cup\ \{ \op{head}(v) \in \history{h} \mid \op{ins}(v) \in \op{I}_\history{h}(i) \}$       & & \\[5pt]
Coherence & memory location     & $k$         & $\{ \op{ins}(v) \in \history{h} \mid v \in V \}$ & $\{ \op{head}(v) \in \history{h} \mid v \in V \}$ & SC     & yes \\[5pt]
PRAM      & thread              & $n$         & $\{ \op{ins}(v) \in \history{h} \mid v \in V \}$ & $\{ \op{head}(v) \in \history{h}|T_i \mid v \in V \}$ & SC     & yes\\[5pt]
PC        & thread              & $n$         & $\{ \op{ins}(v) \in \history{h} \mid v \in V \}$ & $\{ \op{head}(v) \in \history{h}|T_i \mid v \in V \}$ & SC$^a$ & yes \\[5pt]
CC        & thread              & $n$         & $\{ \op{ins}(v) \in \history{h} \mid v \in V \}$ & $\{ \op{head}(v) \in \history{h}|T_i \mid v \in V \}$ & SC$^b$ & yes\\[5pt]
LC        & thread \&  memory location          & $n \cdot k$ & $\{ \op{ins}(v) \in \history{h}|T_i \mid v \in V \}$ & $\{ \op{head}(v) \in \history{h}|T_i \mid v \in V \}$ & SC$^c$ & yes\\
          &     &             & $\cup\ \{ \op{ins}(v) \in \history{h} \mid \op{head}(v) \in \history{h}|T_i \}$    & &\\[2pt]
\bottomrule
\end{tabular}
\begin{flushleft}
SC$^a$: SC and $\op{ins}$-operations are in the same order for each witness.\\
SC$^b$: SC and $\op{ins}$-operations are ordered by the transitive closure of the thread program orders and write-read pairs.\\
SC$^c$: SC and $\op{ins}$-operations from threads other than~$T_i$ can be reordered even if they are from the same thread and only logical contradictions in the local history are considered for consistency.\\
$n$: number of threads, 
$k$: number of memory locations, \\
\textbf{\#SHs}: number of subhistories, 
\textbf{CCfSH}: consistency condition for subhistories, 
\textbf{LoD}: loss of data
\end{flushleft}
\caption{Comparison of consistency conditions for a single distributed shared memory location, i.e., $k = 1$}
\label{tbl:ll-vs-dsmcc}
\end{table}
\end{landscape}

In Table~\ref{tbl:ll-vs-dsmcc} we compare local~linearizability~(LL) against
the consistency conditions coherence~\cite{Ahamad:SPAA93}, 
pipelined~RAM~(PRAM) consistency~\cite{Lipton:TR88,Lipton:patent93,Steinke:JACM04,Ahamad:DC95},
processor~consistency~(PC)~\cite{Ahamad:SPAA93,Goodman:TR91},
causal~consistency~(CC)~\cite{Ahamad:DC95}, and 
local~consistency~(LC)~\cite{Heddaya:TR92}.
Local linearizability shares with all these consistency conditions the idea of 
decomposing a concurrent history into several subhistories.

Coherence projects a concurrent history to the operations on a \emph{single 
memory location} and each resulting history has to be sequentially consistent.
Since sequential consistency is not compositional, coherence does not imply
sequential consistency for the overall history~\cite{Ahamad:SPAA93} whereas 
local linearizability for each single memory location implies local 
linearizability for the overall history.

In contrast to coherence and local consistency, local linearizability, PRAM 
consistency, PC, and CC all decompose the history into per-thread subhistories, i.e., 
if there are $n$ threads then these conditions consider $n$ subhistories and need $n$ sequential 
witnesses.
Coherence requires one witness per memory location and local 
consistency requires one witness per thread and memory location.

For determining the subhistory for a thread~$T_i$, coherence, PRAM consistency, 
PC, and CC consider all write-operations in a given history, i.e.,
$\op{I}_\history{h}(i) = \{ \op{ins}(v) \in \history{h} \mid v \in V \}$.
In contrast, local linearizability only considers the write-operations in 
thread~$T_i$, i.e., $\op{I}_\history{h}(i) = \{ \op{ins}(v) \in \history{h}|T_i \mid v \in V \}$ and local consistency considers
all write-operations in thread~$T_i$ as well as all write-operations whose 
values are read in thread~$T_i$, i.e., $\op{I}_\history{h}(i) = 
\{ \op{ins}(v) \in \history{h}|T_i \mid v \in V \} \cup \{ \op{ins}(v) \in \history{h} \mid \op{head}(v) \in \history{h}|T_i \}$.
Regarding read-operations, PRAM consistency, PC, CC, and LC consider only the
read-operations in thread~$T_i$.
Coherence considers all read-operations in a given history and local 
linearizability only considers read-operations that read the initial 
value~$v_\mathit{init}$ and read-operations that read values that were written 
by a write-operation in thread~$T_i$. Reading the initial value is analogous to returning empty in a data structure.

Local linearizability requires that each subhistory, i.e., thread-induced 
history, is linearizable with respect to the sequential specification under 
consideration.
In contrast, coherence, PRAM consistency, PC, CC, and LC require that each subhistory
is sequentially consistent (or a variant thereof) with respect to the 
sequential specification.
However, the variants of sequential consistency that are used by these 
consistency conditions are vulnerable to a loss of data as discussed in 
Section~\ref{sec:LL-SC-comparison} and, therefore, make these consistency 
conditions unsuitable for concurrent data structures.

When considering PRAM consistency, the sequentialization of the 
write-operations of different threads might be observed differently by 
different threads, e.g., a thread~$T_1$ might observe all write operations of 
thread~$T_2$ before the write operations of thread~$T_3$ but a thread~$T_4$ 
might observe all write operations of $T_3$ before the write operations of 
$T_2$.
In contrast, thread-induced histories as defined by local linearizability do 
not involve write-operations from other threads but involve (some) 
read-operations performed by other threads.
Like PRAM consistency, processor consistency requires for each thread~$T_i$ 
that the read- and write-operations performed by $T_i$ are seen in $T_i$'s 
program order and that the write-operations performed by other threads are 
seen in their respective program order.
Furthermore, processor consistency also requires that two write-operations to 
the \emph{same memory location} appear in the same order in each sequential 
witness of each thread even if they are from different 
threads~\cite{Ahamad:SPAA93,Goodman:TR91}.
This additional condition makes processor consistency strictly stronger than 
PRAM consistency~\cite{Ahamad:SPAA93}.
This condition also creates a similar effect as the consideration of
read-operations in different threads when forming the thread-induced history
in local linearizability.
Causal consistency considers a \emph{causal order} 
instead of the thread program orders alone.
Like local linearizability, causal consistency matches write-read pairs 
across different threads.
In particular, the causal order is the transitive closure of the thread 
program orders and write-read pairs.
By considering the causal order, writes from different threads can become 
ordered which is not the case for local linearizability.

\section{LLD and LL$^+$D Implementation Details} \label{sec:impl-details}

As already mentioned, each thread
inserts elements into a local backend and removes elements either from its local
backend (preferred) or from other backends (fall-back) accessed through a single
segment (thread-indexed array), effectively managing single-producer/multiple-consumer backends
for a varying number of threads.

The segment is dynamic in length (with a
predefined maximum). A slot in this segment
refers to a {\it node} that consists of a backend and a flag
indicating whether the corresponding thread is alive or has terminated. Similar
to other work~\cite{Afek:OPODIS10,Heller:OPODIS05} the flag is used for
logically removing the node from the segment (it stays in the segment until its backend
is empty). Additionally, a (global) version number keeps track of all changes in
the segment. The algorithm is divided into two parts: (1) maintaining
the segment, and (2) adding and removing elements to backends.

In the following we refer to the segment as $s$, a thread's~$T_i$ local node as $n_i$,
the version number of the segment as $v$ and the current length of the segment
as $\ell$. The range of indices $r$ is then defined as $0 \le r < \ell$.

\sloppypar{For maintaining the segment we provide two methods {\tt announce\_thread()}
and {\tt cleanup\_thread(node)} that are used to add and remove nodes to the
segment. Upon removal of a
node the segment is also compacted, i.e., the hole that is created by removing a
node pointer is filled with the last node pointer in the segment. As nodes are
added and removed the length of the segment $\ell$ and thus the range of valid
indices $i$ of the segment, $0 \le i < \ell$, is updated. All changes to the
segment involve incrementing the version number.}

More detailed, the operations for maintaining the segment and compacting it as
nodes are cleaned up are:

\newcommand{\announce}[1]{{\tt announce\_thread(#1)\xspace}}
\newcommand{\cleanup}[1]{{\tt cleanup\_thread(#1)\xspace}}

\tikzset {
	slot/.style={draw, minimum width=0.5cm, minimum height=0.5cm},
	thick slot/.style={draw, minimum width=0.5cm, minimum height=0.5cm, thick},
	n-s/.style={draw, minimum width=1.2cm, minimum height=0.4cm}
}

\newcommand{\drawSlots}[2]{
	\foreach \i in {#1,...,#2} { \node (p\i) at (.5*\i, 0) [thick slot] {}; }
	\draw[->] (0, .5) -- (0, .25) node[above, yshift=.23cm] {\scriptsize $0$};
	\draw[->] (.5*#2, .5) -- (.5*#2, .25) node[above, yshift=.2cm] {\scriptsize $\ell-1$};
	\draw[dashed] (1.75, .25) -- (3, .25) -- (3, -.25) -- (1.75, -.25);
}
\newcommand{\drawDottedArrows}[2]{
	\foreach \i in {#1,...,#2} { \draw[->, dotted] (.5*\i, 0) -- (.5*\i, -.5); }
}
\newcommand{\drawState}[2]{
	\node (p) at (1.5, 1.1) {\scriptsize $\ell = #1$, $v = #2$};
  \node (segment) at (-.5, 0) {\scriptsize $s$:};
}

\begin{itemize}
	\item \announce{}: Allocates a node for the thread as follows: searches for an existing node of a terminated thread and reuses it if it finds one; otherwise it creates a new node, adds the node to~$s$, and adjusts $\ell$.
	In both cases it then increments $v$ and returns the node. The creation of new node is illustrated in Figure~\ref{fig:announce-thread}.

\begin{figure*}
\centering
\begin{minipage}{0.3\textwidth}
\centering
\begin{tikzpicture}
  \drawState{4}{4}
  \drawSlots{0}{3}
  \drawDottedArrows{1}{3}

	\draw[->] (0, 0) -- (0, -.7);
	\node (nodelabel) at (-.9, -1.1) {\scriptsize $n$:};

	\node (n1) at (0, -.9) [n-s] {\tiny $backend$};
	\node (n2) at (0, -1.3) [n-s] {\tiny $alive$};

	\draw[->] (2.75, .5) -- (2.75, .25) node[above, yshift=.23cm] {\scriptsize $max$};

\end{tikzpicture}
\begin{center}(a) Initial state\end{center}
\end{minipage}
\hfill
\begin{minipage}{0.3\textwidth}
\centering
\begin{tikzpicture}
  \drawState{4}{4}
  \drawSlots{0}{3}
  \drawDottedArrows{0}{3}

	\node (s3) at (2, 0) [slot] {};
  \draw[->] (2, 0) -- (2, -.7);

  \node[draw, minimum width=1.2cm, minimum height=0.8cm] (n3) at (2, -1.1)  {};
 	\node (nodelabel) at (1.1, -1.1) {\scriptsize $n_i$:};

\end{tikzpicture}
\begin{center}(b) Add new node $n_i$\end{center}
\end{minipage}
\hfill
\begin{minipage}{0.3\textwidth}
\centering
\begin{tikzpicture}
  \drawState{5}{5}
  \drawSlots{0}{4}
  \drawDottedArrows{0}{3}

  \draw[->] (2, 0) -- (2, -.7);
  \node[draw, minimum width=1.2cm, minimum height=0.8cm] (n3) at (2, -1.1)  {};
 	\node (nodelabel) at (1.1, -1.1) {\scriptsize $n_i$:};

\end{tikzpicture}
\begin{center}(c) Adjust $\ell$, then $v$\end{center}
\end{minipage}
\caption{Segment modifications throughout \announce{}.}
\label{fig:announce-thread}
\vspace{-1em}
\end{figure*}

\item \cleanup{Node n}: Searches for the node $n$ in $s$ using linear search.
  If it finds $n$ at slot $j$, it copies the pointer of $s[\ell -1]$ to $s[j]$,
  decrements $\ell$, increments $v$, and resets $s[\ell]$ to {\tt null} using the new $\ell$. If
  $n$ is not found, then a concurrent thread has already performed the
  cleanup and the operation just returns.  Figure~\ref{fig:cleanup-thread}
  illustrates an example where initially $\ell=5$, the thread owning the node
  at $s[0]$ is dead and the corresponding backend is empty.

\begin{figure*}
\centering
\begin{minipage}{0.3\textwidth}
\centering
\begin{tikzpicture}
  \drawState{5}{5}
  \drawSlots{0}{4}
  \drawDottedArrows{1}{4}

	\draw[->] (0, 0) -- (0, -.7);
	\node (n1) at (0, -.9) [n-s] {\tiny $backend$};
	\node (n2) at (0, -1.3) [n-s] {\tiny $dead$};
	\node (nodelabel) at (-.9, -1.1) {\scriptsize $n$:};

\end{tikzpicture}
\begin{center}(a) $n$ at $s[0]$ is empty and dead\end{center}
\end{minipage}
\hfill
\begin{minipage}{0.3\textwidth}
\centering
\begin{tikzpicture}
  \drawState{5}{5}
  \drawSlots{0}{4}
  \drawDottedArrows{1}{3}

	\draw[->] (0, 0) to[out=-90, in=180, looseness=.5] (1.2, -.6) to[out=0, in=130,looseness=1.5] (2, -.7);
	\node (n1) at (0, -.9) [n-s] {\tiny $backend$};
	\node (n2) at (0, -1.3) [n-s] {\tiny $dead$};
	\node (nodelabel) at (-.9, -1.1) {\scriptsize $n$:};

	\node (s3) at (2, 0) [thick slot] {};
  \draw[->] (2, 0) -- (2, -.7);
  \node[draw, minimum width=1.2cm, minimum height=0.8cm] (n3) at (2, -1.1)  {};
 	\node (nodelabel2) at (1.1, -1.1) {\scriptsize $m$:};

\end{tikzpicture}
\begin{center}(b) Write $s[\ell-1]$ into $s[0]$\end{center}
\end{minipage}
\hfill
\begin{minipage}{0.3\textwidth}
\centering
\begin{tikzpicture}
  \drawState{4}{6}
  \drawSlots{0}{3}
  \drawDottedArrows{1}{3}

	\draw[->] (0, 0) -- (0, -.7);
	\node[draw, minimum width=1.2cm, minimum height=0.8cm] (n3) at (0, -1.1)  {};
	\node (nodelabel) at (-.9, -1.1) {\scriptsize $m$:};

	\node (unusedslot) at (2, 0) [slot] {};

\end{tikzpicture}
\begin{center}Adjust $\ell$ then $v$\end{center}
\end{minipage}
\caption{Segment modifications throughout \cleanup{Node n}.}
\label{fig:cleanup-thread}
\vspace{-1em}
\end{figure*}

\end{itemize}

Note that updating the segment state is only needed when threads are joining or
when backends of terminated threads become empty. We consider both scenarios as infrequent and implement the
corresponding operations using locks. Alternatively those operations can be
implemented using helping approaches, similar to wait-free
algorithms~\cite{Kogan:PPOPP11}. Also note that although operations on segments
are protected by locks, partial changes can be observed, e.g., a remove
operation (as defined below) can observe a segment in an intermediate state with
two pointers pointing to a node during cleanup. The invariant is that no change
can destroy the integrity of the segment within the valid range, i.e., all slots
within the range either point to a valid node or nothing ({\tt null}).

The actual algorithm for adding and removing elements is then defined as
follows:
\begin{itemize}
	\item \ins{}: Upon first insertion, a thread $T_i$ gets assigned a node $n_i$
	(containing backend $b_i$) using \announce{}. The element is then inserted into $b_i$. Subsequent
	insertions from this thread will use $n_i$ throughout the lifetime of the
	thread.
	\item \rem{}: The remove operation consists of two parts: (a) finding
	and removing an element and (b) cleaning up nodes of terminated threads. For
	(a) a thread $T_i$ tries to get an element from its own backend in $n_i$. If
	$n_i$ does not exist (because the thread has not yet performed a single \ins{}
	operation) or the corresponding backend is empty, then a different node $n$ is
	selected randomly within the valid range. If the backend contained in $n$ is
	empty, the operation scans all other nodes' backends in linear fashion.
	However, if the version number changed during the round of scanning through
	all backends, the operation is restarted immediately. 
	Note that since $\ell$ is dynamic a remove operation may operate on a range that is
  no longer valid. Checking the version number ensures that the operation is
  restarted in such a case. 
	For (b) a thread calls
	{\tt cleanup\_thread($n$)} upon encountering a node $n$ that has its alive-flag set to
	false (dead) and contains an empty backend. A cleanup also triggers a restart of
	the remove operation.
	\item ${\tt terminate()}$: Upon termination a thread $T_i$ changes the alive
	flag of $n_i$ to false (dead).
\end{itemize}

Dynamic memory used for nodes is susceptible to the ABA problem and requires
proper handling to free memory. Our implementations use 16-bit ABA counters to
avoid the ABA problem and refrain from freeing memory. Hazard
pointers~\cite{Michael:TPDS04} can be used for solving the ABA problem as well
as for freeing memory.

\subsection{LL+D: LLD with Linearizable Emptiness Check}
\label{sec:llplusd}

We call a data structure implementation $\Phi$ \emph{stateful} if the remove methods of $\Phi$ can be modified to return a so-called \emph{state} that changes upon an insert or a remove of an element, but does not change between two removes that return empty unless an element has been inserted in the data structure in the meantime.  
For stateful implementations $\Phi$ we can create the locally linearizable version with linearizable emptiness check LL$^+$D $\Phi$. 
Michael-Scott queue~\cite{Michael:PODC96} and Treiber 
stack~\cite{Treiber86}  are stateful implementations, whereas LCRQ~\cite{Morrison:PPOPP13} is not.
Also TS stack~\cite{Dodds:POPL15}, and $k$-FIFO~\cite{Kirsch:PACT13} and 
$k$-Stack~\cite{Henzinger:POPL13} are stateful implementations, but the notion of a state in these data structures is huge 
making it unsuitable for LL$^+$D.

For LL$^+$D implementations, linearizable emptiness checks are achieved via an atomic snapshot~\cite{Herlihy:AMP08}, just like for DQs. A detailed description of the LLD and LL$^+$D implementations, as well as the pseudo code, can be found in the appendix. 
Here, we only present the results of the experimental performance evaluation.

\subsection{Correctness of LL$^+$D}\label{sec:LL+D-corr}

\begin{proposition}[LLD and LL$^+$D]
\label{prop:LLD-and-LL$^+$D}
Let $\Phi$ be a stateful data structure implementation that is linearizable with respect to a sequential specification $S_\Phi$. Then LL$^+$D $\Phi$ is linearizable with respect to a pool.
\end{proposition}

\begin{proof} 
Proving that LL$^+$D $\Phi$ is linearizable with respect to pool, in
particular that it has a linearizable emptiness check, follows the proof for
DQ in general, see~\cite{Haas:CF13}: The emptiness check is performed by creating an
atomic snapshot~\cite{Herlihy:AMP08} of the states of all backends (stored in
the {\tt states} array) using the first loop (lines~\ref{lst:dyn-dq-it1-start}-\ref{lst:dyn-dq-it1-end}). If the atomic snapshot is valid (checked via the second loop, lines~\ref{lst:dyn-dq-it2-start}-\ref{lst:dyn-dq-it2-end}, in particular
line~\ref{lst:dyn-dq-state-check}) and all backends are empty in this atomic
snapshot, then there existed a point in time during the creation of the atomic
snapshot where all backends were indeed empty. 

Notice that since the segment is dynamic in length it can happen that some
backends are not contained in the atomic snapshot.  To guarantee that no
elements are missed in the emptiness check the atomic snapshot is extended by
the version number $v$ of the segment. If a new backend is added to the segment
during the generation of the atomic snapshot, then the version number is
increased and the atomic snapshot becomes invalid
(line~\ref{lst:dyn-dq-version-check}).

The linearization point of the remove operation that returns empty is inbetween the two loops
(the last remove attempt of the first loop) if the version check and second loop go through.
\end{proof}

\subsection{LLD Pseudo Code}

All implementations use the interfaces depicted in
Listing~\ref{lst:interfaces}. For simplicity, the interface only mentions pool, queue, and stack.  The highlighted code refers to linearizable emptiness check, i.e., it is only part of the LL$^+$D implementations:  Methods
retrieving elements (e.g. {\tt rem}) are assumed (or modified when possible) to also return a {\tt State}
object that uniquely identifies the state of the data structure with respect to
methods inserting elements (e.g. {\tt ins}). The same state can be accessed via
the {\tt get\_state()} observer method.

\begin{figure}[tb]
\begin{lstlisting}[caption={Pool, queue, and stack interfaces},captionpos=b,
                   label={lst:interfaces},frame=tb]
Pool {
	<Element(*\scriptsize\ttfamily\colorbox{light-gray}{, State}*)> rem();
	void ins(Element e);
	(*\scriptsize\ttfamily\colorbox{light-gray}{State get\_state();}*)
}

Queue : Pool {
	<Element(*\scriptsize\ttfamily\colorbox{light-gray}{, State}*)> dequeue();
	void enqueue(Element e);
	void ins(Element e) => enqueue(e);
	<Element(*\scriptsize\ttfamily\colorbox{light-gray}{, State}*)> rem() => dequeue();
}

Stack : Pool {
	<Element(*\scriptsize\ttfamily\colorbox{light-gray}{, State}*)> pop();
	void push(Element e);
	void ins(Element e) => push(e);
	<Element(*\scriptsize\ttfamily\colorbox{light-gray}{, State}*)> rem() => pop();
}
\end{lstlisting}
\end{figure}

Listing~\ref{lst:dynamic-loc-lin-dq-1} illustrates the pseudo-code for
maintaining the segment. The backend on
line~\ref{lst:dyn-dq-backend} can either be declared as {\tt Stack} or {\tt
Queue} as defined in Listing~\ref{lst:interfaces} (or any other linearizable data structure).

\begin{figure}[tb]
\begin{lstlisting}[caption=Node and segment structure for LLD and LL$^+$D (queue or stack),captionpos=b,
									 label={lst:dynamic-loc-lin-dq-1},frame=tb]
Node {
 (* \label{lst:dyn-dq-backend} *)Pool backend;  // Any linearizable data structure.
	Bool alive;
}

Segment {
	Node nodes[MAX_THREADS];
	Int l = 0;
	Int version = 0;

	// Returns all indexes between 0 and l (exclusive) in random order.
	[Int] range();

	// Announces a node in the buffer, effectively adding it to nodes_,
	// adjusting l, and changing the version.
	Node announce_thread() {
		segment_lock(); // Protecting against concurrent announce or cleanup operations.
    Node n = find_dead_node();
    if n == null {
      n = Node(b: Backend());
    	nodes[l] = n;
		  l++;
    }
    n.alive = true;
		version++;
		segment_unlock();
    return n;
	}

	// Removes a node from the buffer, effectively removing it from nodes_,
	// adjusting l, and changing the version.
	void cleanup_thread(Node n, Int old_version) {
		segment_lock(); // Protecting against concurrent announce or cleanup operations.
		<j, error> = find_node_in_segment(n);
		if error || n.alive || old_version != version {
			segment_unlock();
			return;
		}
		nodes[j] = nodes[l-1];
		l--;
		version++;
		nodes[l] = null;
		segment_unlock();
	}
}
\end{lstlisting}
\end{figure}

Listing~\ref{lst:dynamic-loc-lin-dq-2} shows the pseudo-code for LL$^+$D.
When removing the highlighted code, we obtain the code for LLD.
Each thread maintains its own backend, enclosed in a thread-local node
(line~\ref{lst:dyn-dq-local-node}), for insertion. The local backend is always
accessed through {\tt get\_local\_node} (line~\ref{lst:dyn-dq-get-node}). This
method also makes sure that a thread is
 announced
(line~\ref{lst:dyn-dq-announce}) upon first insertion and acquires a node. An \ins{} operation then
always uses a thread's local backend
(line~\ref{lst:dyn-dq-get-node-call}~and~\ref{lst:dyn-dq-put-local}) for
insertion. For removing an element in \rem{}, a thread tries to remove an
element from its local backend first
(line~\ref{lst:dyn-dq-get-1}-\ref{lst:dyn-dq-get-2}). If no element can be
found, all backends in the valid range are searched in a linear fashion,
starting from a random index.
The highlighted code (lines~\ref{lst:dyn-dq-it2-start}-\ref{lst:dyn-dq-it2-end})
illustrates checking the atomic snapshot for LL$^+$D.

\begin{figure}[tb!]
\begin{lstlisting}[frame=t]
DynamicLocallyLinearizableDQ {
  Segment s;
  thread_local Node local_node;(*\label{lst:dyn-dq-local-node}*)

  Node get_local_node(Bool create_if_absent) {(*\label{lst:dyn-dq-get-node}*)
	  if (create_if_absent) && (local_node == null) {(*\label{lst:dyn-dq-check-node}*)
      local_node = s.announce_thread();(*\label{lst:dyn-dq-announce}*)
		}
		return local_node;
	}

  void ins(Element e) {
    n = get_local_node(create_if_absent: true);(*\label{lst:dyn-dq-get-node-call}*)
    n.backend.ins(e);(*\label{lst:dyn-dq-put-local}*)
  }

  Element rem() {
  	// Fast path of retrieving an element from the thread-local backend.
  	n = get_local_node(create_if_absent: false);(*\label{lst:dyn-dq-get-1}*)
  	if n != null {
  		<e(*\scriptsize\ttfamily\colorbox{light-gray}{, state}*)> = n.backend.rem();
  		if e != null { return e; }
  	}(*\label{lst:dyn-dq-get-2}*)
  	while true {
  		retry = false;
	  	old_version = s.version;
	  	range = s.range();(*\label{lst:dyn-dq-range}*)
	  	for i in range {(*\label{lst:dyn-dq-it1-start}*)
	  		n = s.nodes[i];
        if old_version != s.version { 
          retry = true; break; }
        Bool alive = n.alive;
	  		<e(*\scriptsize\ttfamily\colorbox{light-gray}{, state}*)> = n.backend.rem();
	  		if e == null {
	  			(*\scriptsize\ttfamily\colorbox{light-gray}{states[i] = state;}*)
	  			if !alive { 
            s.cleanup_thread(n, old_version); 
            retry = true; break; 
          }
	  		} else {
	  			return e;
	  		}
	  	}(*\label{lst:dyn-dq-it1-end}*)
	  	if retry { continue; } 
	  	if old_version != s.version {  continue; }(*\label{lst:dyn-dq-version-check}*)
\end{lstlisting}
\vspace{-.4cm}
\begin{lstlisting}[firstnumber=last,backgroundcolor=\color{light-gray}]{}
	  	for i in range {(*\label{lst:dyn-dq-it2-start}*)
	  		n = s.nodes[i];
	  		if n == null || n.backend.get_state() != states[i] {(*\label{lst:dyn-dq-state-check}*)
	  			retry = true;
	  			break;
	  		}
	  	}(*\label{lst:dyn-dq-it2-end}*)
	  	if retry { continue; }
\end{lstlisting}
\vspace{-.4cm}
\begin{lstlisting}[firstnumber=last,caption={LLD and LL$^+$D (queue and stack)},captionpos=b,
									 label={lst:dynamic-loc-lin-dq-2},frame=b]{}

	  	return null; // Empty case.
  	}
  }

  // Called upon thread termination.
  void terminate() { 
  	 n = get_local_node(create_if_absent: false);
  	 if n != null { n.alive = false; }
  }
}
\end{lstlisting}
\end{figure}

\subsection{LLD with Observer Methods}
\label{sec:lld_and_observers}

We have implemented LLD variants of (strict and relaxed) queue and stack implementations. None of our LLD implementations involves observer methods, but the LLD algorithm can easily be extended to support observer methods:
\begin{itemize}
	\item A data observer on LLD $\Phi$ (independently of which thread performs it) amounts to a data observer on any $\Phi_T$.
	\item A local shape observer on LLD $\Phi$ performed by thread $T$ executes the shape observer on~$\Phi_T$. 
	\item A global shape observer on LLD $\Phi$ executes the shape observer on each backend $\Phi_T$ and produces an aggregate value. 
\end{itemize}

\section{Additional Implementations}\label{sec:implementations}

We now present and evaluate additional algorithms that provide locally linearizable
variants of queues and stacks, obtained by modifying relaxed $k$-out of order
queues and stacks~\cite{Henzinger:POPL13,Kirsch:PACT13} in a way that makes them sequentially correct.
We have also tried another generic implementation, related to the construction in~\cite{Cerone:ICALP14}, that implements a flat-combining  wrapper with sequential (to be precise, single-producer multiple-consumer) backends. 
In our initial experiments the performance of such an implementation was not particularly promising. 

\subsection{Locally Linearizable $k$-FIFO Queue and $k$-Stack}

$k$-FIFO queues~\cite{Kirsch:PACT13} and $k$-Stacks~\cite{Henzinger:POPL13} are
relaxed queues and stacks based on lists of segments where each segment holds
$k$ slots for elements, effectively allowing reorderings of elements of up to $k
- 1$. The list of segments is implemented by a variant of Michael-Scott
queue~\cite{Michael:PODC96} for $k$-FIFO and a variant of Treiber
stack~\cite{Treiber86} for $k$-Stack. Insert and remove methods operate on the
segments ignoring any order of elements within the same segment. Segments used
for insertion and removal are identified by insertion and removal pointers,
respectively.

For queues, elements are removed from the oldest segment and inserted into the
most-recent not-full segment. Upon trying to remove an element from an empty
segment the segment is removed and the removal pointer advanced to the next
segment. Upon trying to insert an element into a full segment a new segment is
appended and the insertion pointer is advanced to this new segment. Similarly
(but different) for a stack, removal and insertion operate on the most-recent
segment, i.e., removal and insertion pointer are synonyms and identify the same
segment at all times. Again, upon trying to remove an element from an empty
segment the segment is removed and the removal pointer advanced to the next
segment. Upon trying to insert an element into a full segment a new segment is
prepended and the insertion pointer is set to this new segment.

$k$-FIFO queues and $k$-Stacks are relaxed queues and stacks that are: (1)
linearizable with respect to $k$-out-of-order queue and
stack~\cite{Henzinger:POPL13}, respectively; (2) linearizable with respect to a
pool~\cite{Henzinger:POPL13,Kirsch:PACT13}; (3) not locally linearizable with
respect to queue and stack, respectively, for $k\ge1$ since reordering elements
that are inserted in the same segment (even sequentially by a single thread) is
allowed, see the histories (b) and (c) in the proof of
Theorem~\ref{thm:sc:queues-and-stacks}; and (4) not sequentially consistent with
respect to queue and stack, as shown by the histories (d) and (e) in the proof
of Theorem~\ref{thm:sc:queues-and-stacks} that are $k$-FIFO and $k$-Stack
histories, respectively, for $k\ge1$.

We now present LL $k$-FIFO and LL $k$-Stack, modifications of $k$-FIFO and
$k$-Stack, that enforce local linearizability by ensuring that no thread inserts
more than once in a single segment. Assuming that segments are unique (by
tagging pointers), LL $k$-FIFO remembers the last used insertion pointer per
thread. For LL $k$-Stack the situation is more subtle as (due to the stack
semantics) segments can be reached multiple times for insertion and removal.
Figure~\ref{fig:k-stack-dilemma} illustrates an example where the top segment of
a $k$-Stack is reached multiple times by the same thread~($T_1$). Since in the
general case all segments could be reached multiple times by a single thread it
is required to maintain the full history of each thread's insertions. Assuming
the maximum number of threads is known in advance, a bitmap is used to maintain
the information in which segment a thread has already pushed a value. One can
similarly implement a locally linearizable version of the
Segment~Queue~\cite{Afek:OPODIS10}.

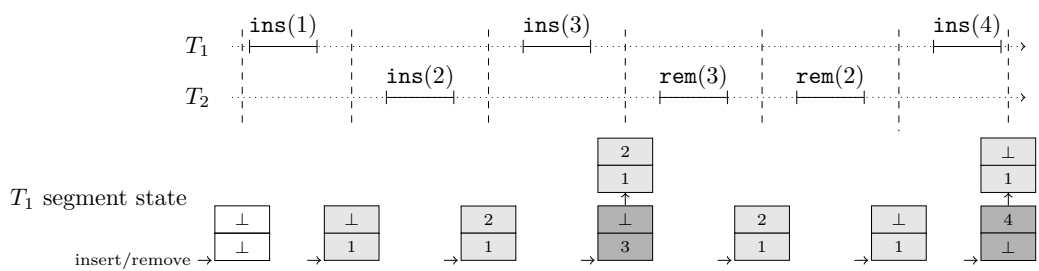
\begin{figure*}[ht!]
\centering
\scalebox{.9}{
\begin{tikzpicture}

  \tikzset {
		kslot/.style={draw, minimum width=0.8cm, minimum height=0.4cm},
		kslot-filled/.style={kslot,fill=black!10},
		kslot-filled-dark/.style={kslot,fill=black!30}
  }
	
	\node (T1) at (0.25,.75) {$T_1$};
	\node (T2) at (0.25,0) {$T_2$};
	\node (T1_marking) at (-1.2, -1.5) {$T_1$ segment state};

	\draw[dotted,->] (0.75, .75) -- (12.35, .75);
	\draw[dotted,->] (0.75, 0) -- (12.35, 0);

	\draw[|-|] (1, .75) -- node[above] {$\ins{1}$} (2, .75);
	\draw[|-|] (3, 0) -- node[above] {$\ins{2}$} (4, 0);
	\draw[|-|] (5, .75) -- node[above] {$\ins{3}$} (6, .75);
	\draw[|-|] (7, 0) -- node[above] {$\rem{3}$} (8, 0);
	\draw[|-|] (9, 0) -- node[above] {$\rem{2}$} (10, 0);
	\draw[|-|] (11, .75) -- node[above] {$\ins{4}$} (12, .75);
	
	\draw[dashed] (0.9, 1) -- (0.9, -0.4);
	\node (state1c) at (.9, -1.8) [kslot] {\scriptsize $\bot$};
	\node (state1d) at (.9, -2.2) [kslot] {\scriptsize $\bot$};
	\node (top) at (-0.7, -2.4) [align=center] {\scriptsize insert/remove};
	\draw[->] (.25,-2.4) -- (.45,-2.4);

	\draw[dashed] (2.5, 1) -- (2.5, -0.4);
	\node (state2c) at (2.5, -1.8) [kslot-filled] {\scriptsize $\bot$};
	\node (state2d) at (2.5, -2.2) [kslot-filled] {\scriptsize $1$};
	\draw[->] (1.85, -2.4) -- (2.05, -2.4);

	\draw[dashed] (4.5, 1) -- (4.5, -0.4);
	\node (state3c) at (4.5, -1.8) [kslot-filled] {\scriptsize $2$};
	\node (state3d) at (4.5, -2.2) [kslot-filled] {\scriptsize $1$};
	\draw[->] (3.85, -2.4) -- (4.05, -2.4);

	\draw[dashed] (6.5, 1) -- (6.5, -0.4);
	\node (state4a) at (6.5, -0.8) [kslot-filled] {\scriptsize $2$};
	\node (state4b) at (6.5, -1.2) [kslot-filled] {\scriptsize $1$};
	\node (state4c) at (6.5, -1.8) [kslot-filled-dark] {\scriptsize $\bot$};
	\node (state4d) at (6.5, -2.2) [kslot-filled-dark] {\scriptsize $3$};
	\draw[->] (state4c) -- (state4b);
	\draw[->] (5.85, -2.4) -- (6.05, -2.4);

	\draw[dashed] (8.5, 1) -- (8.5, -0.4);
	\node (state5c) at (8.5, -1.8) [kslot-filled] {\scriptsize $2$};
	\node (state5d) at (8.5, -2.2) [kslot-filled] {\scriptsize $1$};
	\draw[->] (7.85, -2.4) -- (8.05, -2.4);

	\draw[dashed] (10.5, 1) -- (10.5, -0.5);
	\node (state6c) at (10.5, -1.8) [kslot-filled] {\scriptsize $\bot$};
	\node (state6d) at (10.5, -2.2) [kslot-filled] {\scriptsize $1$};
	\draw[->] (9.85, -2.4) -- (10.05, -2.4);

	\draw[dashed] (12.1, 1) -- (12.1, -0.4);
	\node (state7a) at (12.1, -0.8) [kslot-filled] {\scriptsize $\bot$};
	\node (state7b) at (12.1, -1.2) [kslot-filled] {\scriptsize $1$};
	\node (state7c) at (12.1, -1.8) [kslot-filled-dark] {\scriptsize $4$};
	\node (state7d) at (12.1, -2.2) [kslot-filled-dark] {\scriptsize $\bot$};
	\draw[->] (state7c) -- (state7b);
	\draw[->] (11.45, -2.4) -- (11.65, -2.4);

\end{tikzpicture}
}
\caption{LL $k$-Stack run ($k = 2$). $T_1$ can only insert in uncolored segments and 
         needs to prepend a new segment (for insertion) otherwise.}
\label{fig:k-stack-dilemma}
\end{figure*}

\subsubsection{$k$-FIFO Queue and LL $k$-FIFO Queue Pseudo Code.}

Listing~\ref{lst:ll-kfifo} shows the pseudo code for LL $k$-FIFO queue. Again we
highlight the code we added to the original pseudo code~\cite{Kirsch:PACT13}.
Similar to the locally linearizable $k$-Stack each thread inserts at most one
element into a segment. However, in the k-FIFO queue we do not need flags in
each segment to achieve this property. It is sufficient  to remember the last
segment used for insertion for each thread ({\tt set\_last\_tail};
line~\ref{lst:ll-kfifo-set}).  For each enqueue the algorithm checks whether the
executing thread has already used this segment for enqueueing an element ({\tt
get\_last\_tail}; line~\ref{lst:ll-kfifo-check}). If the segment has already
been used, the thread tries to append a new segment (effectively adding a new
tail).

\begin{figure}[h!]
\begin{lstlisting}[frame=t]
LocallyLinearizableKFIFOQueue {                
	enqueue(item):
	  while true:
	    tail_old = get_tail();
\end{lstlisting}
\vspace{-.4cm}
\begin{lstlisting}[firstnumber=last,backgroundcolor=\color{light-gray}]{}
			(* \label{lst:ll-kfifo-check} *)if get_last_tail(thread_id) == tail_old:
				advance_tail(tail_old, k);
				continue;  // Restart while loop.
\end{lstlisting}
\vspace{-.4cm}
\begin{lstlisting}[firstnumber=last]{}
	    head_old = get_head();
	    item_old, index = find_empty_slot(tail_old, k);
	    if tail_old == get_tail():
	      if item_old.value == EMPTY:
	        item_new = atomic_value(item, item_old.version + 1);
	        if CAS(&tail_old->segment[index], item_old, item_new):
	          if committed(tail_old, item_new, index):
\end{lstlisting}
\vspace{-.4cm}
\begin{lstlisting}[firstnumber=last,backgroundcolor=\color{light-gray}]{}
						(* \label{lst:ll-kfifo-set} *)set_last_tail(thread_id, tail_old);
\end{lstlisting}
\vspace{-.4cm}
\begin{lstlisting}[firstnumber=last,frame=b,captionpos=b,caption={Locally Linearizable $k$-FIFO Queue},label={lst:ll-kfifo}]{}
	            return true;
	      else:
	        advance_tail(tail_old, k);

	bool committed(tail_old, item_new, index):
	  if tail_old->segment[index] != item_new:
	    return true;
	  head_current = get_head();
	  tail_current = get_tail();
	  item_empty = atomic_value(EMPTY, item_new.version + 1);
	  if in_queue_after_head(tail_old, tail_current, head_current):
	    return true;
	  else if not_in_queue(tail_old, tail_current, head_current):
	    if !CAS(&tail_old->segment[index], item_new, item_empty):
	      return true;
	  else: //in queue at head
	    head_new = atomic_value(head_current.value, head_current.version + 1);
	    if CAS(&head, head_current, head_new):
	      return true;
	    if !CAS(&tail_old->segment[index], item_new, item_empty):
	      return true;
	  return false;

	item dequeue():
	  while true:
	    head_old = get_head();
	    item_old, index = find_item(head_old, k);
	    tail_old = get_tail();
	    if head_old == get_head():
	      if item_old.value != EMPTY:
	        if head_old.value == tail_old.value:
	          advance_tail(tail_old, k);
	        item_empty = atomic_value(EMPTY, item_old.version + 1);
	        if CAS(&head_old[index], item_old, item_empty):
	          return item_old.value;
	      else:
	        if head_old.value == tail_old.value && tail_old.value == get_tail():
	          return null;
	        advance_head(head_old, k);
}
\end{lstlisting}
\end{figure}

\subsubsection{Correctness Proof of LL $k$-FIFO Queue.}

Having Theorem~\ref{thm:queue}, the proof of correctness of LL $k$-FIFO queue is easy.

\begin{theorem}[Correctness of LL $k$-FIFO]
LL $k$-FIFO queue presented in Listing~\ref{lst:ll-kstack} is locally linearizable.
\end{theorem}

\begin{proof}
Using Theorem~\ref{thm:queue},
 as a first proof obligation
we have to show that any history~$\history{h}$ of the LL $k$-FIFO queue is
locally linearizable with respect to the pool sequential specification
$S_P$. This proof is analogous to the proof that any history
of the LL $k$-Stack is locally linearizable with respect to the pool
sequential specification $S_P$, and is therefore postponed until the corresponding LL $k$-Stack theorem.

What remains to show is that
\begin{center}
\begin{tabular}{@{}l@{}}
$\forall x,y \in V. \,\,\,\forall i. \,\,\,\enq{x} \precOrder{\history{h}}^i
  \enq{y} \,\,\,\wedge\,\,\, \deq{y} \in \history{h}
  \,\,\,\Rightarrow\,\,\, \deq{x} \in \history{h} \,\,\,\wedge\,\,\,\deq{y}
  \not<_{\history{h}} \deq{x}$
\end{tabular}
\end{center}

  Assume $\enq{x} \precOrder{\history{h}}^i \enq{y}$. This means that $x$ and $y$ were
  enqueued by the same thread $i$ and therefore inserted into different
  segments. Moreover, the segment of $x$ is closer to the head of the list
  than the segment of $y$. A $\deq{y}$ method call can
  remove $y$ only if the segment of $y$ is the head segment. The segment of
  $y$ can only become the head segment if all segments closer to the head
  of the list get empty. This means that also the segment of $x$ has to become
  empty. Therefore there has to exist a $\deq{x}$ method call which removes
  $x$ from the segment, and $\deq{x} \not<_{\history{h}} \deq{y}$.

\end{proof}

\subsubsection{$k$-Stack and LL $k$-Stack Pseudo Code.}

Listing~\ref{lst:ll-kstack} shows the pseudo code for LL $k$-Stack. The
highlighted code is the code we added to the original pseudo
code~\cite{Henzinger:POPL13} to achieve local linearizability. The difference to
the original algorithm is that a thread inserts at most one element into a
segment. To achieve this property each segment in the k-stack contains a flag
per thread which is set when an element is inserted into the segment ({\tt
mark\_segment\_as\_used}; line~\ref{lst:ll-kstack-mark1} and
line~\ref{lst:ll-kstack-mark2}). If a thread encounters a segment where its flag
is already set, the thread does not insert its element into that segment but
tries to prepend a new segment ({\tt is\_segment\_marked};
line~\ref{lst:ll-kstack-is-marked}). Otherwise the element is inserted into the
existing segment and the flag of the thread in that segment is set.

\begin{lstlisting}[frame=t]
LocallyLinearizableKStack {
	SegmentPtr top;

	void init():
	  new_ksegment = calloc(sizeof(ksegment));
	  top = atomic_value(new_ksegment, 0);

	bool try_add_new_ksegment(top_old, item):
	  if top_old == top:
	    new_ksegment = calloc(sizeof(ksegment));
	    new_ksegment->next = top_old;
	    new_ksegment->s[0] = atomic_value(item, 0); // Use first slot for item.
	    top_new = atomic_value(new_ksegment, top_old.ver+1);
\end{lstlisting}
\vspace{-.4cm}
\begin{lstlisting}[firstnumber=last,backgroundcolor=\color{light-gray}]{}
			(* \label{lst:ll-kstack-mark1} *)mark_segment_as_used(top_new);
\end{lstlisting}
\vspace{-.4cm}
\begin{lstlisting}[firstnumber=last]{}
	    if CAS(&top, top_old, top_new):(*\label{lineks:cas_segment}*)
	      return true;
	  return false;

	void try_remove_ksegment(top_old):
	  if top_old == top:
	    if top_old->next != null:
	      atomic_increment(&top_old->remove);
	      if empty(top_old):
	        top_new = atomic_value(top_old->next, top_old.ver+1);
	        if CAS(&top, top_old, top_new):
	          return;
	      atomic_decrement(&top_old->remove);

	bool committed(top_old, item_new, index):
	  if top_old->s[index] != item_new:
	    return true;
	  else if top_old->remove == 0:
	    return true;
	  else:  //top_old->remove >= 1
	    item_empty = atomic_value(EMPTY, item_new.ver+1);
	    if top_old != top:
	      if !CAS(&top_old->s[index], item_new, item_empty):
	        return true;
	    else:
	      top_new = atomic_value(top_old.val, top_old.ver+1);
	      if CAS(&top, top_old, top_new):
	        return true;
	      if !CAS(&top_old->s[index], item_new, item_empty):
	        return true;
	  return false;

	void push(item):
	  while true:
	    top_old = top;
\end{lstlisting}
\vspace{-.4cm}
\begin{lstlisting}[firstnumber=last,backgroundcolor=\color{light-gray}]{}
			(* \label{lst:ll-kstack-is-marked} *)if segment_is_marked(top_old):
				if try_add_new_ksegment(top_old, item);
					return true;
				continue;  // Restart while loop.
\end{lstlisting}
\vspace{-.4cm}
\begin{lstlisting}[firstnumber=last]{}
	    item_old, index = find_empty_slot(top_old);(*\label{lineks:find_empty_slot}*)
	    if top_old == top:
	      if item_old.val == EMPTY:
	        item_new = atomic_value(item, item_old.ver+1);
	        if CAS(&top_old->s[index], item_old, item_new):(*\label{lineks:cas_element}*)
	          if committed(top_old, item_new, index):(*\label{lineks:commit}*)
\end{lstlisting}
\vspace{-.4cm}
\begin{lstlisting}[firstnumber=last,backgroundcolor=\color{light-gray}]{}
							(* \label{lst:ll-kstack-mark2} *)mark_segment_as_used(old_top);
\end{lstlisting}
\vspace{-.4cm}
\begin{lstlisting}[firstnumber=last,caption={Locally Linearizable $k$-Stack},
                   label={lst:ll-kstack},captionpos=b,frame=b]{}
	            return true;
	      else:
	        if try_add_new_ksegment(top_old, item):
	          return true;

	item pop():
	  while true:
	    top_old = top;
	    item_old, index = find_item(top_old);(*\label{lineks:find_item}*)
	    if top_old == top:(*\label{lineks:top_check_pop}*)
	      if item_old.val != EMPTY:
	        item_empty = atomic_value(EMPTY, item_old.ver+1);
	        if CAS(&top_old->s[index], item_old, item_empty):(*\label{lineks:cas_pop}*)
	          return item_old.val;
	      else:
	        if only_ksegment(top_old):
	          if empty(top_old):(*\label{lineks:empty}*)
	            if top_old == top:
	              return null;
	        else:
	          try_remove_ksegment(top_old);
}
\end{lstlisting}

\subsubsection{Correctness Proof of LL $k$-Stack.}

The local linearizability proof of LL $k$-Stack is more involved, but very interesting. We use a theorem from
the published artifact of~\cite{Dodds:POPL15}, which has been mechanically proved in
the Isabelle HOL theorem prover. 

\begin{theorem}[Empty Returns for Stack]
\label{thm:emptiness}
  Let $\history{h}$ be a history, and let $\history{h'}$ be the projection
  of $\history{h}$ to $\Sigma \setminus \pop{\emptyValue}$. If
  $\history{h}$ is linearizable with respect to the sequential
  specification $S_P$ of a pool (see Definition~\ref{example:pool}), and
  $\history{h'}$ is linearizable with respect to the sequential
  specification $S_S$ of a stack (see Definition~\ref{example:stack}), then
  $\history{h}$ is linearizable with respect to~$S_S$.
\end{theorem}

\begin{proof}
  Here we repeat the key insights of the proof and leave out
  technical details. A complete and mechanized version of the proof is
  available in the published artifact of~\cite{Dodds:POPL15}.

  As $\history{h}$ is linearizable with respect to $S_P$, and
  $\history{h'}$ is linearizable with respect to $S_S$, there exists a
  sequential history $\history{s} \in S_P$ such that $\history{s}$ is a
  linearization of $\history{h}$, and there exists a sequential history
  $\history{s'} \in S_S$ such that $\history{s'}$ is a linearization of
  $\history{h'}$.  We show that we can construct a sequential history
  $\history{t} \in S_S$ such that $\history{t}$ is a linearization of~$\history{h}$.

  The linearization $\history{t}$ is constructed as follows: the position
  of $\pop{\emptyValue}$ in $\history{s}$ is preserved in $\history{t}$.
  This means for any method call $m \in \history{s}$ that if
  $\pop{\emptyValue} \abOrder{\history{s}} m$, then also $\pop{\emptyValue}
  \abOrder{\history{t}} m$, and if $m \abOrder{\history{s}} \pop{\emptyValue}$, then also $m
  \abOrder{\history{t}} \pop{\emptyValue}$.  Moreover, if two method calls $m, n \in
  \history{s}$ are ordered as $m \abOrder{\history{s}} \pop{\emptyValue} \abOrder{\history{s}} n$
  and therefore by transitivity it holds that $m \abOrder{\history{s}} n$, then also
  $m \abOrder{\history{t}} n$. 
  
  For all other method calls the order of $\history{s'}$ is preserved. This
  means for any two method calls $m, n \in \history{s}$ with $m \abOrder{\history{s'}}
  n$, that if for all $\pop{\emptyValue}$ it holds that $\pop{\emptyValue}
  \abOrder{\history{s}} m$ if and only if $\pop{\emptyValue} \abOrder{\history{s}} n$, then $m
  \abOrder{\history{t}} n$.

  By construction, the history $\history{t}$ is sequential and a permutation
  of $\history{h}$. Next we show that $\history{t}$ is a linearization of
  $\history{h}$ by showing that $\history{t}$ preserves the precedence
  order of $\history{h}$. Also by construction, it holds that if $m
  \abOrder{\history{t}} n$ for any two method calls $m, n \in \history{t}$, then also
  either $m \abOrder{\history{s}} n$ or $m \abOrder{\history{s'}} n$. Both $\history{s}$ and
  $\history{s'}$ are linearizations of $\history{h}$ and $\history{h'}$,
  respectively. Therefore it cannot be for any $m, n$
  with $m \precOrder{\history{h}} n$ that $n \abOrder{\history{s}} m$ or $n \abOrder{\history{s'}} m$, it
  can also not be that $n \abOrder{\history{t}} m$. Since $\history{t}$ is
  sequential, this means that $\history{t}$ preserves the precedence order
  of $\history{h}$.
 
  Next we show that $\history{t} \in S_P$ according to
  Definition~\ref{example:pool}:

  \begin{enumerate}[(1)]

    \item Every method call, but $\pop{\emptyValue}$, appears in
      $\history{s}$ at most once: This is guaranteed since $\history{t}$ is
      a permutation of $\history{s}$, and $\history{s} \in S_P$.
    
    \item If $\pop{x}$ appears in $\history{t}$, then also $\push{x}$ does
      and $\push{x} \abOrder{\history{t}} \pop{x}$: again, since $\history{t}$ is a
      permutation of $\history{s}$ and $\history{s} \in S_P$, if $\pop{x}
      \in \history{t}$, then also $\push{x} \in \history{t}$. Since
      $\push{x} \abOrder{\history{s}} \pop{x}$ and $\push{x} \abOrder{\history{s'}} \pop{x}$
      (because both $\history{s}$ and $\history{s'}$ are in $S_P$) it
      also holds that $\push{x} \abOrder{\history{t}} \pop{x}$, as we argued already
      above. 
   
    \item $\forall x \in V.\,\,\, \push{x}
      \abOrder{\history{t}}\pop{\emptyValue} \Rightarrow \pop{x}
      \abOrder{\history{t}}\pop{\emptyValue}$: this property is satisfied
      trivially as all $\pop{\emptyValue}$ operations are ordered the same
      in $\history{t}$ as in $\history{s}$, and $\history{s} \in S_P$.
   
  \end{enumerate}

  It only remains to check that all elements are removed in a stack
  fashion. We have to show the following:
\begin{center}
\begin{tabular}{@{}l@{}}
$\forall x,y \in V.\,\, \push{x} \abOrder{\history{t}} \push{y} \abOrder{\history{t}} \pop{x}
  \,\,\,\Rightarrow \,\,\, \pop{y} \in \history{t} \,\, \wedge \,\, \pop{y}
  \abOrder{\history{t}} \pop{x}$
\end{tabular}
\end{center}

  First we show that if $\push{x} \abOrder{\history{t}} \push{y} \abOrder{\history{t}}
  \pop{x}$, then also $\push{x} \abOrder{\history{s'}} \push{y} \abOrder{\history{s'}}
  \pop{x}$. We do this by showing that there cannot exist a
  $\pop{\emptyValue}$ such that $\push{x} \abOrder{\history{t}} \pop{\emptyValue}
  \abOrder{\history{t}} \push{y}$ or $\push{y} \abOrder{\history{t}} \pop{\emptyValue}
  \abOrder{\history{t}} \pop{x}$.

  Assume, towards a contradiction, $\push{x} \abOrder{\history{t}}
  \pop{\emptyValue} \abOrder{\history{t}} \push{y}$. By the transitivity of
  $\abOrder{\history{t}}$ this implies that $\push{x} \abOrder{\history{t}} \pop{\emptyValue}
  \abOrder{\history{t}} \pop{x}$, which contradicts our observation above that
  $\history{t} \in S_P$. Therefore $\push{x} \abOrder{\history{t}} \pop{\emptyValue}
  \abOrder{t} \push{y}$ is not possible, and for the same reason also
  $\push{y} \abOrder{\history{t}} \pop{\emptyValue} \abOrder{\history{t}} \pop{x}$ is not
  possible.

  Now, as $\history{s'} \in S_S$ and $\push{x} \abOrder{\history{s'}} \push{y}
  \abOrder{\history{s'}} \pop{x}$, there has to exist a
  $\pop{y} \in \history{s'}$ with $\pop{y} \abOrder{\history{s'}} \pop{x}$.  For the
  same reason as above it cannot be that $\pop{x} \abOrder{\history{s}}
  \pop{\emptyValue} \abOrder{\history{s}} \pop{y}$.  Therefore $\pop{y}$ and
  $\pop{x}$ are ordered in $\history{t}$ the same as in $\history{s'}$,
  i.e. $\pop{y} \abOrder{\history{t}} \pop{x}$, and therefore $\history{t} \in S_S$. 

\end{proof}

\begin{theorem}[Correctness of LL $k$-Stack]
  The LL $k$-Stack algorithm presented in Listing~\ref{lst:ll-kstack} is locally
  linearizable.
\end{theorem}

\begin{proof}
  We have to show that every history $\history{h}$ of LL $k$-Stack is
  locally linearizable with respect to the sequential specification $S_S$
  defined in Definition~\ref{example:stack}. This means that we have to show
  that every thread-induced history $\history{h}_i$ of $\history{h}$ is
  linearizable with respect to $S_S$ for any thread $i$.

  Having Theorem~\ref{thm:emptiness}
  we only have to show that
  $\history{h}_i$ is linearizable with respect to the sequential
  specification $S_P$ of a pool (defined in Definition~\ref{example:pool}),
  and that $\history{h'}_i$, the projection of $\history{h}_i$ to $\Sigma
  \; \setminus \; \pop{\emptyValue}$, is linearizable with respect to the
  sequential specification~$S_S$ of a stack.

  We start with the proof that $\history{h}_i$ is linearizable with respect
  to $S_P$. 
  We construct a sequential history~$\history{s}_i$ from
  $\history{h}_i$ by identifying the linearization points of the push and
  pop method calls of the LL $k$-Stack. 
  This means that two method calls $m,n$ are ordered in $\history{s}_i$, 
  $m \abOrder{\history{s}_i} n$ if the linearization point of $m$ is executed before
  the linearization point of $n$ in $\history{h}_i$.

  The linearization point of push method calls is either the successful
  insertion of a new segment in line~\ref{lineks:cas_segment}, or the last successful {\tt
  CAS} which writes the element into a segment slot in line~\ref{lineks:cas_element}. The
  linearization point of pop method calls is the successful {\tt CAS} which
  removes an element from its segment slot in line~\ref{lineks:cas_pop}. 
  
  For the linearization point of $\pop{\emptyValue}$ we take the
  linearization point of the call to {\tt empty} in line~\ref{lineks:empty}. The {\tt
  empty} method creates an atomic snapshot~\cite{Herlihy:AMP08} of the
  top segment.  This atomic snapshot is the state of the top segment at
  some point (i.e. linearization point of {\tt empty}) within the execution
  of {\tt empty}. If {\tt empty} returns {\tt true}, then there exists no
  element in the atomic snapshot of the segment.

  Next we show that $\history{s}_i$ is in $S_P$ as defined in
  Definition~\ref{example:pool}.

\begin{enumerate}[(1)]

  \item Since there exists exactly one linearization point per method call,
    every method call, but $\rem{\emptyValue}$, appears in $\history{s}_i$
    at most once.
    
   \item  If $\pop{x}$ appears in $\history{s}_i$, then it reads $x$ in a
     slot of the top segment before its linearization point. Since only push
     method calls write their elements into segment slots, there has to exist
     a $\push{x}$ which wrote $x$ into that slot. Therefore the
     linearization point of $\push{x}$ is always before the linearization
     point of $\pop{x}$, and therefore $\push{x} \abOrder{\history{s}_i} \pop{x}$.
   
   \item Segments are only removed from the list of segments when they
     become empty. The call to {\tt committed} guarantees that elements are
     not inserted into segments which are about to be removed.
     
     A pop method calls {\tt empty} only if there is a single segment
     left in the LL $k$-Stack and no element was found in that segment in {\tt
     find\_item}. 
     
     Now assume a $\push{x}$ method call inserts an element $x$ which is
     missed by {\tt find\_item}. If $\push{x}$ wrote $x$ into a segment
     before the linearization point of $\pop{\emptyValue}$ and the segment
     was not the last segment, then the top segment changed since
     $\pop{\emptyValue}$ searched for an element and therefore the check in
     line~\ref{lineks:top_check_pop} would fail. If $\push{x}$ wrote $x$
     into the last segment of the LL $k$-Stack, then a $\pop{x}$ method call
     removed $x$ from the segment because otherwise $x$ would be in the
     atomic snapshot of {\tt empty} and therefore {\tt empty} would return
     {\tt false}. Therefore, if $\push{x} \abOrder{\history{s}_i} \pop{\emptyValue}$,
     then also $\pop{x} \abOrder{\history{s}_i} \pop{\emptyValue}$. 
     
\end{enumerate}
  
  \noindent Therefore $\history{s}_i$ is in the sequential specification~$S_P$ of a
  pool. 

  Next we show that $\history{h'}_i$ is linearizable with respect
  to~$S_S$. 
  We construct again a sequential history $\history{s'}_i$ from
  $\history{h'}_i$ by identifying the linearization points of the push and
  pop method calls of LL $k$-Stack.

  The linearization point of the push operations is the successful
  insertion of a new segment in line~\ref{lineks:cas_segment} if it is executed, or the
  reading of the empty slot (line~\ref{lineks:find_empty_slot}) in the last (and therefore
  successful) iteration of the main loop. The linearization point of a pop
  operation is the reading of a non-empty slot (line~\ref{lineks:find_item}) in the last
  (and therefore successful) iteration of the main loop. There do not exist
  any $\pop{\emptyValue}$ method calls in $\history{s'}_i$. Since we assume a
  sequentially consistent memory model, these read operations define a total
  order on the LL $k$-Stack method calls in $\history{h'}_i$.

  First we show that $\history{s'}_i$ is in the sequential specification
  $S_P$ of a pool as defined in Definition~\ref{example:pool}.

\begin{enumerate}[(1)]

  \item Since there exists exactly one linearization point per method call,
    every method call appears in $\history{s'}_i$
    at most once.
    
   \item  If $\pop{x}$ appears in $\history{s'}_i$, then it read $x$ in a
     slot of the top segment at its linearization point. Since only push
     operations write their elements into segment slots, there has to exist
     a $\push{x}$ which wrote $x$ into that slot. The linearization point of 
     $\push{x}$ is always before $x$ is written into a segment slot.
     Therefore $\push{x} \abOrder{\history{s'}_i} \pop{x}$.
   
   \item Since there exist no $\pop{\emptyValue}$ operations in
     $\history{s'}_i$ the third pool condition is trivially correct.
   
\end{enumerate}

  Next we show that $\history{s'}_i$ also provides a stack order, which
  means that we have to show that 
\begin{center}
\begin{tabular}{@{}l@{}}
$\forall x,y \in V.\,\, \push{x} \abOrder{\history{s'}_i} \push{y} \abOrder{\history{s'}_i}
  \pop{x} \,\,\,\Rightarrow \,\,\, \pop{y} \in \history{s'}_i \,\, \wedge
  \,\, \pop{y} \abOrder{\history{s'}_i} \pop{x}.$
\end{tabular}
\end{center}

\noindent
We start by observing some invariants. 

\begin{enumerate}

	\item  A thread never inserts
  	elements into the same segment twice. This is guaranteed by the call to
  	{\tt segment\_is\_marked}.
  
  	\item Between the linearization point of a push and the time it writes
  	its element into a segment the segment the element gets written into is
  	not removed: if the push operation inserts a new segment this is
  	trivially correct. If the push operation writes the element into an
  	existing segment, then the call to {\tt committed} in
  	line~\ref{lineks:commit} guarantees that the segment  was not removed.

  	\item At the time of the linearization point of the pop, which is the
  	time when the pop reads the non-empty slot (line~\ref{lineks:find_item})
  	in the last (and therefore successful) iteration, the pop reads the
  	non-empty slot from the top segment. This is guaranteed by the check in
  	line~\ref{lineks:top_check_pop}. 
\end{enumerate}

  Now assume there exist the operations $\push{x}$, $\push{y}$ and
  $\pop{x}$ in $\history{s'}_i$ and $\push{x} \abOrder{\history{s'}_i} \push{y}
  \abOrder{\history{s'}_i} \pop{x}$. Since $\push{x}$ and $\push{y}$ are both in
  $\history{s'}_i$, this means that both operations are executed by the same
  thread. Therefore, according to Invariant 1., $x$ and $y$ get
  inserted into different segments, with the segment $y$ on top of the
  segment of $x$. 
  
  The linearization point of $\pop{x}$ cannot be before $y$ is written into
  its segment because according to Invariant 2. the segment $y$ gets
  inserted into does not get removed between the linearization point of
  $\push{y}$ and the time $y$ is written into the segment. With Invariant 3.
  this means that $x$ is unaccessible for $\pop{x}$ before $y$
  gets written into a segment. Also because of the third invariant the top
  segment changes between the insertion of $y$ and the linearization point
  of $\pop{x}$.

  Next we observe that as long as $y$ is not removed, no segment below the
  segment of $y$ can become the top segment. Therefore for the segment of
  $x$ to become the top segment so that $\pop{x}$ can remove it, $y$ has to
  be removed first. Only a $\pop{y}$ can remove $y$, and therefore there
  exists a $\pop{y}$ and the linearization point of $\pop{y}$ is before the
  linearization point of $\pop{x}$.
  
  Hence $\history{s'_i}$ is in the sequential specification of a stack.
  Using Theorem~\ref{thm:emptiness}
  this means that LL $k$-Stack in
  listing~\ref{lst:ll-kstack} is locally-linearizable with respect to the
  sequential specification of a stack.

\end{proof}

\section{Additional Experiments}

We also evaluate the implementations on another Scal workload, the sequential alternating workload. However, we note that in this workload in the locally linearizable implementations threads only access their local backends, so no wonder they perform perfectly well.

\subparagraph{Mixed Workload.}

In order to evaluate the performance and scalability of mixed workloads, i.e.,
workloads where threads produce and consume values, we exercise the so-called
sequential alternating workload in Scal. Each thread is configured to execute
$10^6$ pairs of insert and remove operations, i.e., each insert operation is
followed by a remove operation. As in the producer-consumer workload, the
contention is controlled by adding a busy wait of $5\mu s$. 
The number of threads is configured to range between $1$ and $80$. 
Again we report the number of data structure operations per second.

Data structures that require parameters to be set are configured like in
the producer-consumer benchmark.
Figure~\ref{fig:seqalt} shows the results of the mixed workload benchmark for
all considered data structures.  

The MS queue and Treiber stack do not perform and scale for more than
10 threads. As in the producer-consumer benchmark, LCRQ and TS Stack either perform competitively with
their $k$-out-of-order relaxed counter parts $k$-FIFO and $k$-Stack or even
outperform and outscale them (in the case of LCRQ, that even outperforms the pool). 

LL$^+$D MS queue, LLD LCRQ, and LL$^+$D Treiber stack perform very well and scale (nearly) linearly in the number of threads.
A surprising result is that LLD $k$-FIFO performs poorly in this experiment. The reason is that $k$-FIFO performs poorly when it is almost empty, and in this experiment
each backend instance of LLD $k$-FIFO contains at most one element at any point in time. 
The $k$-Stack  performs better on a nearly-empty state. 
The benefit of trying to perform a local operation first in the LLD algorithms is
visible when comparing to 1-RA DQ and DS that do not utilize a local fast path. 

\begin{figure}[t]
\begin{minipage}{.5\textwidth}
	\centering
	\resizebox{.95\textwidth}{!}{
		\begingroup
		  \makeatletter
		  \gdef\gplbacktext{}%
		  \gdef\gplfronttext{}%
		  \makeatother
		      \def\colorrgb#1{\color[rgb]{#1}}%
		      \def\colorgray#1{\color[gray]{#1}}%
		      \expandafter\def\csname LTw\endcsname{\color{white}}%
		      \expandafter\def\csname LTb\endcsname{\color{black}}%
		      \expandafter\def\csname LTa\endcsname{\color{black}}%
		      \expandafter\def\csname LT0\endcsname{\color[rgb]{1,0,0}}%
		      \expandafter\def\csname LT1\endcsname{\color[rgb]{0,1,0}}%
		      \expandafter\def\csname LT2\endcsname{\color[rgb]{0,0,1}}%
		      \expandafter\def\csname LT3\endcsname{\color[rgb]{1,0,1}}%
		      \expandafter\def\csname LT4\endcsname{\color[rgb]{0,1,1}}%
		      \expandafter\def\csname LT5\endcsname{\color[rgb]{1,1,0}}%
		      \expandafter\def\csname LT6\endcsname{\color[rgb]{0,0,0}}%
		      \expandafter\def\csname LT7\endcsname{\color[rgb]{1,0.3,0}}%
		      \expandafter\def\csname LT8\endcsname{\color[rgb]{0.5,0.5,0.5}}%
		  \setlength{\unitlength}{0.0500bp}%
		  \begin{picture}(7200.00,5040.00)%
		    \gplbacktext{%
		      \csname LTb\endcsname%
		      \put(814,1584){\makebox(0,0)[r]{\strut{} 0}}%
		      \put(814,1983){\makebox(0,0)[r]{\strut{} 4}}%
		      \put(814,2382){\makebox(0,0)[r]{\strut{} 8}}%
		      \put(814,2781){\makebox(0,0)[r]{\strut{} 12}}%
		      \put(814,3180){\makebox(0,0)[r]{\strut{} 16}}%
		      \put(814,3578){\makebox(0,0)[r]{\strut{} 20}}%
		      \put(814,3977){\makebox(0,0)[r]{\strut{} 24}}%
		      \put(814,4376){\makebox(0,0)[r]{\strut{} 28}}%
		      \put(814,4775){\makebox(0,0)[r]{\strut{} 32}}%
		      \put(946,1364){\makebox(0,0){\strut{}1}}%
		      \put(1605,1364){\makebox(0,0){\strut{}10}}%
		      \put(2337,1364){\makebox(0,0){\strut{}20}}%
		      \put(3069,1364){\makebox(0,0){\strut{}30}}%
		      \put(3801,1364){\makebox(0,0){\strut{}40}}%
		      \put(4533,1364){\makebox(0,0){\strut{}50}}%
		      \put(5266,1364){\makebox(0,0){\strut{}60}}%
		      \put(5998,1364){\makebox(0,0){\strut{}70}}%
		      \put(6730,1364){\makebox(0,0){\strut{}80}}%
		      \put(176,3179){\rotatebox{-270}{\makebox(0,0){\strut{}million operations per sec (more is better)}}}%
		      \put(3874,1034){\makebox(0,0){\strut{}number of threads}}%
		    }%
		    \gplfronttext{%
		      \csname LTb\endcsname%
		      \put(1885,613){\makebox(0,0)[r]{\strut{}MS}}%
		      \csname LTb\endcsname%
		      \put(1885,393){\makebox(0,0)[r]{\strut{}LCRQ}}%
		      \csname LTb\endcsname%
		      \put(1885,173){\makebox(0,0)[r]{\strut{}k-FIFO}}%
		      \csname LTb\endcsname%
		      \put(4091,613){\makebox(0,0)[r]{\strut{}LL+D~MS}}%
		      \csname LTb\endcsname%
		      \put(4091,393){\makebox(0,0)[r]{\strut{}LLD~LCRQ}}%
		      \csname LTb\endcsname%
		      \put(4091,173){\makebox(0,0)[r]{\strut{}LLD~k-FIFO}}%
		      \csname LTb\endcsname%
		      \put(6297,613){\makebox(0,0)[r]{\strut{}1-RA~DQ}}%
		    }%
		    \gplbacktext
		    \put(0,0){\includegraphics{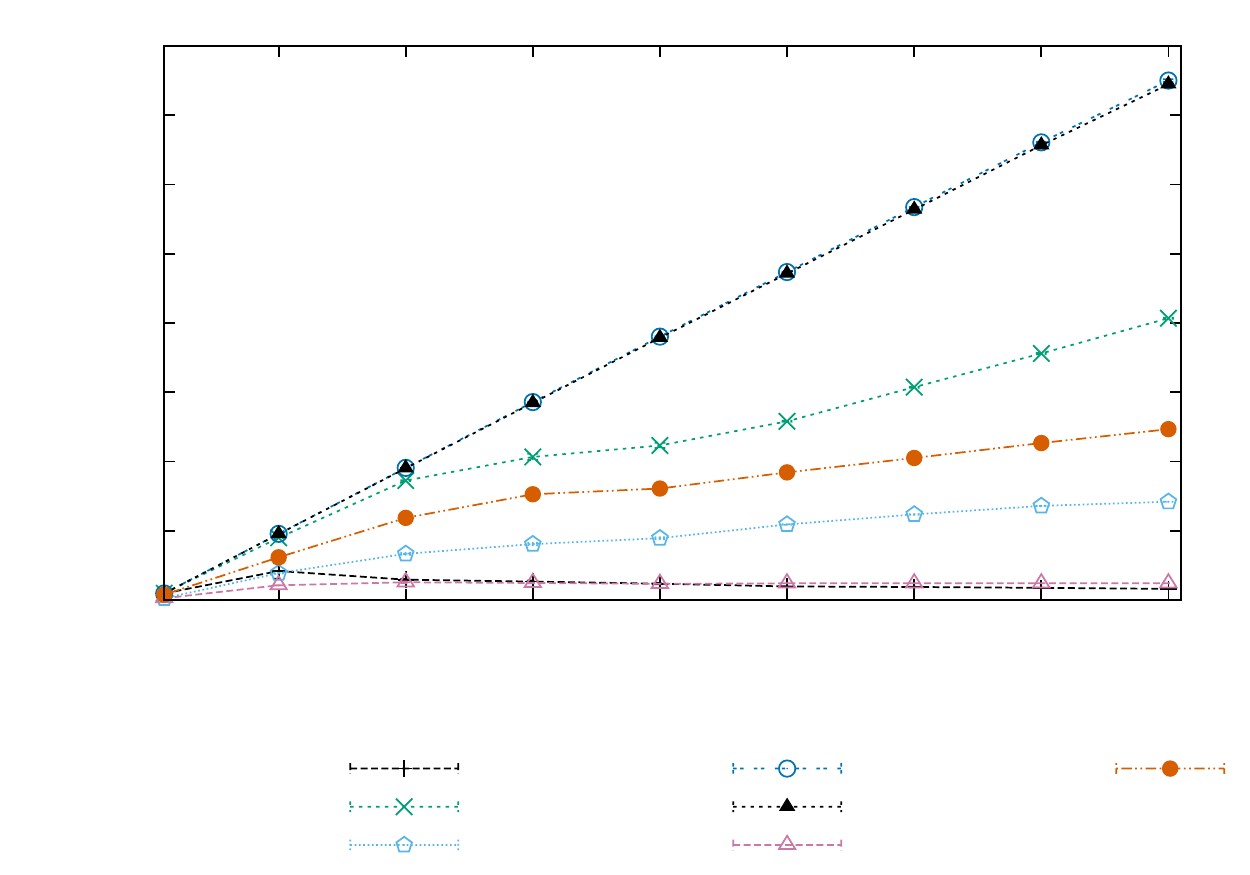}}%
		    \gplfronttext
		  \end{picture}%
		\endgroup
	}
	\begin{center}Queues, LL queues, and ``queue-like'' pools\end{center}
	\label{fig:seqalt-queues}
\end{minipage}
\begin{minipage}{.5\textwidth}
	\centering
	\resizebox{.95\textwidth}{!}{
		\begingroup
		  \makeatletter
		  \gdef\gplbacktext{}%
		  \gdef\gplfronttext{}%
		  \makeatother
		      \def\colorrgb#1{\color[rgb]{#1}}%
		      \def\colorgray#1{\color[gray]{#1}}%
		      \expandafter\def\csname LTw\endcsname{\color{white}}%
		      \expandafter\def\csname LTb\endcsname{\color{black}}%
		      \expandafter\def\csname LTa\endcsname{\color{black}}%
		      \expandafter\def\csname LT0\endcsname{\color[rgb]{1,0,0}}%
		      \expandafter\def\csname LT1\endcsname{\color[rgb]{0,1,0}}%
		      \expandafter\def\csname LT2\endcsname{\color[rgb]{0,0,1}}%
		      \expandafter\def\csname LT3\endcsname{\color[rgb]{1,0,1}}%
		      \expandafter\def\csname LT4\endcsname{\color[rgb]{0,1,1}}%
		      \expandafter\def\csname LT5\endcsname{\color[rgb]{1,1,0}}%
		      \expandafter\def\csname LT6\endcsname{\color[rgb]{0,0,0}}%
		      \expandafter\def\csname LT7\endcsname{\color[rgb]{1,0.3,0}}%
		      \expandafter\def\csname LT8\endcsname{\color[rgb]{0.5,0.5,0.5}}%
		  \setlength{\unitlength}{0.0500bp}%
		  \begin{picture}(7200.00,5040.00)%
		    \gplbacktext{%
		      \csname LTb\endcsname%
		      \put(814,1804){\makebox(0,0)[r]{\strut{} 0}}%
		      \put(814,2175){\makebox(0,0)[r]{\strut{} 4}}%
		      \put(814,2547){\makebox(0,0)[r]{\strut{} 8}}%
		      \put(814,2918){\makebox(0,0)[r]{\strut{} 12}}%
		      \put(814,3290){\makebox(0,0)[r]{\strut{} 16}}%
		      \put(814,3661){\makebox(0,0)[r]{\strut{} 20}}%
		      \put(814,4032){\makebox(0,0)[r]{\strut{} 24}}%
		      \put(814,4404){\makebox(0,0)[r]{\strut{} 28}}%
		      \put(814,4775){\makebox(0,0)[r]{\strut{} 32}}%
		      \put(946,1584){\makebox(0,0){\strut{}1}}%
		      \put(1605,1584){\makebox(0,0){\strut{}10}}%
		      \put(2337,1584){\makebox(0,0){\strut{}20}}%
		      \put(3069,1584){\makebox(0,0){\strut{}30}}%
		      \put(3801,1584){\makebox(0,0){\strut{}40}}%
		      \put(4533,1584){\makebox(0,0){\strut{}50}}%
		      \put(5266,1584){\makebox(0,0){\strut{}60}}%
		      \put(5998,1584){\makebox(0,0){\strut{}70}}%
		      \put(6730,1584){\makebox(0,0){\strut{}80}}%
		      \put(176,3289){\rotatebox{-270}{\makebox(0,0){\strut{}million operations per sec (more is better)}}}%
		      \put(3874,1254){\makebox(0,0){\strut{}number of threads}}%
		    }%
		    \gplfronttext{%
		      \csname LTb\endcsname%
		      \put(2988,833){\makebox(0,0)[r]{\strut{}Treiber}}%
		      \csname LTb\endcsname%
		      \put(2988,613){\makebox(0,0)[r]{\strut{}TS~Stack}}%
		      \csname LTb\endcsname%
		      \put(2988,393){\makebox(0,0)[r]{\strut{}k-Stack}}%
		      \csname LTb\endcsname%
		      \put(2988,173){\makebox(0,0)[r]{\strut{}LL+D~Treiber}}%
		      \csname LTb\endcsname%
		      \put(5458,833){\makebox(0,0)[r]{\strut{}LLD~TS~Stack}}%
		      \csname LTb\endcsname%
		      \put(5458,613){\makebox(0,0)[r]{\strut{}LLD~k-Stack}}%
		      \csname LTb\endcsname%
		      \put(5458,393){\makebox(0,0)[r]{\strut{}1-RA~DS}}%
		    }%
		    \gplbacktext
		    \put(0,0){\includegraphics{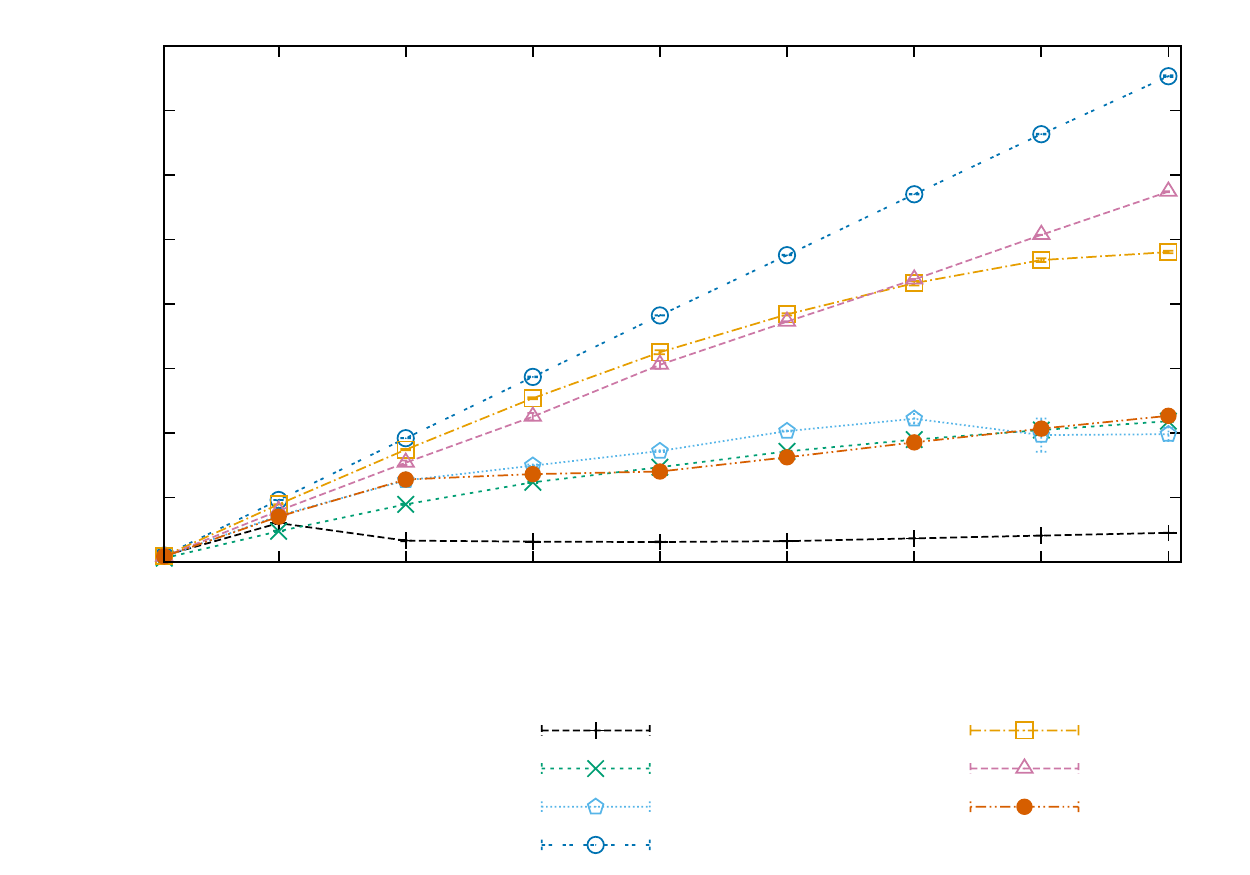}}%
		    \gplfronttext
		  \end{picture}%
		\endgroup
	}
	\begin{center}Stacks, LL stacks, and ``stack-like'' pools\end{center}
	\label{fig:seqalt-stacks}
\end{minipage}
\caption{
Performance and scalability of sequential alternating microbenchmarks with an
increasing number of threads on a 40-core (2 hyperthreads per core) machine}
\label{fig:seqalt}
\end{figure}

\section{Verifying Local Linearizability}

In general, verifying local linearizability amounts to verifying linearizability for a set of smaller histories. 
This might enable verification in a modular/compositional way.
Aside from this, it is important to mention (again) that for our locally linearizable data structures in Section~\ref{sec:implementations-queues-and-stacks} built from linearizable building blocks, the correctness proofs are straightforward assuming the building blocks are proven to be linearizable. 
In addition, for queue we can state an ``axiomatic" verification theorem for local linearizability in the style of~\cite{Henzinger:CONCUR13,Chakraborty:LMCS15}, whose main theorem we recall next (with a slight reformulation).

\begin{theorem}[Queue Linearizability]
\label{thm:queue-lin}
A queue concurrent history $\history{h}$ is linearizable wrt the queue sequential specification~$S_Q$ if and only if 
\begin{enumerate}[1.]

	\item $\history{h}$ is linearizable wrt the pool sequential specification $S_P$ (with suitable renaming of method calls), and
	
	\item $\forall x,y \in V. \,\,\,\enq{x} <_\history{h} \enq{y}\,\,\,\wedge\,\,\, \deq{y} \in \history{h} \,\,\,\Rightarrow\,\,\, \deq{x} \in \history{h} \,\,\,\wedge\,\,\,\deq{y} \not<_\history{h} \deq{x}$.
	\qed
\end{enumerate}
\end{theorem}

We note that an analogous change to the axioms in the sequential specification of a pool and a stack does not lead to a characterisation of linearizability for pools and stacks, cf.~\cite{Dodds:POPL15}. 
An axiomatic characterisation of linearizability for pools and stacks would involve an infinite number of axioms/infinite axioms, due to the need to prohibit infinitely many problematic shapes, cf.~\cite{Bouajjani:ICALP15}.

We are now able to state the queue-local-linearizability-verification result.
\begin{theorem}[Queue Local Linearizability]\label{thm:queue}
A queue concurrent history $\history{h}$ is locally linearizable wrt the queue sequential specification $S_Q$ if and only if 
\begin{enumerate}[1.]

	\item $\history{h}$ is locally linearizable wrt the pool sequential specification $S_P$ (after suitable renaming of method calls), and
	
	\item $\forall x,y \in V. \,\,\,\forall T. \,\,\,\enq{x} <_{\history{h}}^T \enq{y} \,\,\,\wedge\,\,\, \deq{y} \in \history{h} \,\,\,\Rightarrow\,\,\, \deq{x} \in \history{h} \,\,\,\wedge\,\,\,\deq{y} \not<_\history{h} \deq{x}$.
	\qed
\end{enumerate}
\end{theorem}

\end{document}